\documentclass[12pt, reqno]{amsart}
\usepackage{amsxtra}
\usepackage[all]{xy}
\usepackage{amssymb}
\usepackage{amsmath}
\usepackage{tikz-cd}
\usepackage{array}
\usepackage{booktabs}
\usepackage{ifpdf}
\ifpdf
\usepackage[colorlinks,final,backref=page,hyperindex]{hyperref}
\else
\usepackage[colorlinks,final,backref=page,hyperindex,hypertex]{hyperref}
\fi
\usepackage{tikz}
\usepackage[active]{srcltx}
\usepackage{makecell}
\usepackage{float}
\usepackage{multirow}
\theoremstyle{plain}
\newcommand{\cleqn}{\setcounter{equation}{0}}
\newcommand{\clth}{\setcounter{theorem}{0}}
\newcommand {\sectionnew}[1]{\section{#1}\cleqn\clth}

\newtheorem{theorem}{Theorem}[section]
\newtheorem{lemma}[theorem]{Lemma}
\newtheorem{definition-theorem}[theorem]{Definition-Theorem}
\newtheorem{proposition}[theorem]{Proposition}
\newtheorem{corollary}[theorem]{Corollary}
\newtheorem{definition}[theorem]{Definition}
\newtheorem{example}[theorem]{Example}
\newtheorem{remark}[theorem]{Remark}
\newtheorem{notation}[theorem]{Notation}
\newtheorem{assumption}[theorem]{Assumption}
\newtheorem{lemma-definition}[theorem]{Lemma-Definition}
\newtheorem{lemma-notation}[theorem]{Lemma-Notation}
\newtheorem{question}[theorem]{Question}
\newtheorem{remark-definition}[theorem]{Remark-Definition}
\newcommand \bth[1] { \begin{theorem}\label{t#1} }
\newcommand \ble[1] { \begin{lemma}\label{l#1} }

\newcommand \bpr[1] { \begin{proposition}\label{p#1} }
\newcommand \bco[1] { \begin{corollary}\label{c#1} }
\newcommand \bde[1] { \begin{definition}\label{d#1}\rm }
\newcommand \bex[1] { \begin{example}\label{e#1}\rm }
\newcommand \bre[1] { \begin{remark}\label{r#1}\rm }

\newcommand \bnota[1] {\begin{notation}\label{n#1}\rm }
\newcommand \bas[1] { \begin{assumption}\label{a#1}\rm }

\newcommand \bqu[1] { \begin{question}\label{q#1}\rm }

\newcommand {\ele} { \end{lemma} }

\newcommand {\epr} { \end{proposition} }
\newcommand {\eco} { \end{corollary} }
\newcommand {\ede} { \end{definition} }
\newcommand {\eex} { \end{example} }
\newcommand {\ere} { \end{remark} }
\newcommand {\enota} { \end{notation} }
\newcommand {\eas} {\end{assumption}}

\newcommand {\equ} {\end{question}}

\newcommand \lb[1]{\label{#1}}

\def \g  {\mathfrak{g}}   

\def \n  {\mathfrak{n}}

\def \sl {\mathfrak{sl}}
\def \gl {\mathfrak{gl}}

\DeclareMathOperator \ad { {\mathrm{ad}} }

\newcommand{\beqa}{\begin{eqnarray*}}                     
\newcommand{\eeqa}{\end{eqnarray*}}

\def \sC {{\scriptscriptstyle C}}

\def \wF_mn {\wF_m \times \wF_n}
\def \wF_mnC {\wF_{m, n, \, \sC}}

\def \wF {\widetilde{F}}

\newcommand\blfootnote[1]{%
  \begingroup
  \renewcommand\thefootnote{}\footnote{#1}%
  \addtocounter{footnote}{-1}%
  \endgroup
}

\begin{document}

\setlength{\baselineskip}{1.2\baselineskip}
\title[Nijenhuis operators on 2D pre-Lie algebras and 3D associative algebras]
{Nijenhuis operators on 2D pre-Lie algebras and 3D associative algebras}

\author{Xiaoguang Zou}
\address{
Zhejiang college \\
Shanghai University of Finance and Economics\\
Zhejiang, Jinhua 321013 \\
China}
\email{zminger@126.com}

\author{Xiang Gao}
\address{
School of Mathematical Sciences    \\
Zhejiang Normal University\\
Jinhua 321004              \\
China}
\email{gaoxiang@zjnu.edu.cn}

\author{Chuangchuang Kang}
\address{
School of Mathematical Sciences    \\
Zhejiang Normal University\\
Jinhua 321004              \\
China}
\email{kangcc@zjnu.edu.cn}

\author{Jiafeng L\"u}
\address{
School of Mathematical Sciences    \\
Zhejiang Normal University\\
Jinhua 321004              \\
China}
\email{jiafenglv@zjnu.edu.cn}

\blfootnote{*Corresponding Author: Chuangchuang Kang. Email: kangcc@zjnu.edu.cn.}
\date{}
\begin{abstract}In this paper, we describe all Nijenhuis operators on 2-dimensional complex pre-Lie algebras and 3-dimensional complex associative algebras. As an application, using these operators, we obtain solutions of the classical Yang-Baxter equation on the corresponding sub-adjacent Lie algebras.
\end{abstract}



\maketitle
\tableofcontents
\allowdisplaybreaks
\sectionnew{Introduction}\lb{intro}
This paper aims to describe Nijenhuis operators on 2-dimensional complex pre-Lie algebras and 3-dimensional complex associative algebras, and to explore their applications in the solutions of the classical Yang-Baxter equation.
\subsection{Nijenhuis algebras and Nijenhuis torsion}
A \textbf{Nijenhuis algebra} is a triple $(A,\cdot, N)$ including an algebra $(A,\cdot)$ and a Nijenhuis operator $N$ on $A$, i.e., a linear map $N : A \rightarrow  A$ satisfies
\begin{equation}\label{eq:Ni ope1}
N(x) \cdot N(y)=N(N(x) \cdot  y+x \cdot N(y)-N(x \cdot  y)),\quad\forall~ x, y \in A.
\end{equation}
Nijenhuis torsion was first introduced by Nijenhuis in 1951 (\cite{Nijenhuis}).
Let $M$ be a smooth $n$-dimensional manifold, $N$ a smooth operator field on it, and
$v$, $u$ smooth vector fields. Then the Nijenhuis torsion $(\mathcal{T}_N)_{~jk}^{i}$ of $N$ is a tensor of type (1, 2) given by the formula
$$
\mathcal{T}_N(v,u)=[Nv,Nu]+N^2[v,u]-N[Nv,u]-N[v,Nu],
$$
or in local coordinates $\mathbf{x}=(x^1,...,x^n)$,
$$
(\mathcal{T}_N)_{~jk}^{i}=N_{~j}^l\frac{\partial N_{~k}^i}{\partial x^l}-N_{~k}^l\frac{\partial N_{~j}^i}{\partial x^l}-N_{~l}^i\frac{\partial N_{~k}^l}{\partial x^j}+N_{~l}^i\frac{\partial N_{~j}^l}{\partial x^k}.
$$
An operator $N$ is called a Nijenhuis operator if $\mathcal{T}_N = 0$. The vanishing of the Nijenhuis torsion is a necessary (but not sufficient) condition for the integrability of an operator field (\cite{Konyaev}). Moreover, Nijenhuis operators play a central role in various fields such as complex differential geometry and complex Lie algebra theory (\cite{Wells}), deformations of associative (resp. Lie) algebras (\cite{Gerstenhaber,Nijenhuis-2}), and pairs of Dirac structures (\cite{clem,Dorfman}). Many interesting generalizations have been made, including applications to pre-Lie algebras (\cite{Wang}), $n$-Lie algebras (\cite{Liu}), Hom-Lie algebras (\cite{Das}), and Courant algebroids with Nijenhuis structures (\cite{Kos}). The theory of Nijenhuis operators has recently been enriched by studies extending them to broader frameworks, including Banach homogeneous spaces (\cite{Golinski2}), Leibniz algebras(\cite{Mondal}), associative D-bialgebras(\cite{Ma}), and 3-Hom-Lie algebras(\cite{Li}), among others.

\subsection{Nijenhuis operators on pre-Lie algebras}

Pre-Lie algebras, also known as left-symmetric algebras, were first introduced by Cayley in the context of rooted tree algebras (\cite{Cayley}). They have been used in the study of homogeneous convex cones (\cite{Vinberg}), affine manifolds, and affine structures on Lie groups (\cite{Koszul,Medina}). In \cite{Winterhalder}, it was shown that an operator depending linearly on local coordinates is a Nijenhuis operator if and only if its coefficients are structure constants of a pre-Lie algebra. In \cite{Wang}, Wang, Sheng, Bai and Liu defined Nijenhuis operators on pre-Lie algebras, which generate a trivial deformation of pre-Lie algebras. The local behavior of these operators at singular points is important in the study of complex integrable systems (\cite{Bolsinov}). So far, only a few Nijenhuis operators on (simple) pre-Lie algebras are known explicitly, and a full classification of simple pre-Lie algebras has not yet been completed (\cite{Bai-1,Burde,Kim,Kong}). Constructing examples of Nijenhuis operators on low-dimensional pre-Lie algebras can help better understand Nijenhuis operators and related topics.

A Rota-Baxter operator of weight zero on the sub-adjacent Lie algebra can be obtained from a Nijenhuis operator on a pre-Lie algebra.  Let $A$ be a pre-Lie algebra (product of
$x$ and $y$ is denoted by $xy$)  over a field $\mathbb{C}$.
A linear map  $R: A \rightarrow A$~is called a Rota-Baxter operator of weight $\lambda\in  \mathbb{C}$ on $A$ if
\begin{equation}
R(x)R(y)  = R(R(x)y + xR(y)+ \lambda xy),~\forall~ x, y \in A.
\end{equation}
If a Nijenhuis operator $ N $ on a pre-Lie algebra satisfies $ N^2 = 0 $, 
then~$N$~is a Rota-Baxter operator of weight zero on  $A$ (see Proposition \ref{Ni ope and rb ope}).
Let~$\g$~be a Lie algebra. A linear map  $R:\g \rightarrow \g$~
is called an operator form of the {\bf classical Yang-Baxter equation} (CYBE) if,
\begin{equation}\label{eq:RB ope}
[R(x), R(y)] = R([R(x), y] + [x, R(y)]),~\forall~ x, y \in \mathfrak{g},
\end{equation}
which is also called a Rota-Baxter operator on ${\frak g}$. The CYBE originated from the
study of inverse scattering theory, which is recognized as the
``semi-classical
limit'' of the quantum Yang-Baxter equation (\cite{Belavin,Faddeev1,Faddeev2}). If $R$ is a Rota-Baxter operator of weight zero on a pre-Lie algebra $A$, then $R$ is
a Rota-Baxter operator of weight zero on the sub-adjacent Lie algebra ${\frak g}(A)$ (see Proposition \ref{rb on pre-Lie and Lie}).  

Based on the above relationship between Nijenhuis operators and Rota-Baxter operators, solutions of the CYBE on semidirect product sub-adjacent Lie algebras can be derived. Recall that $r=\sum\limits_i a_i\otimes b_i\in {\frak g}\otimes {\frak g}$ is a solution of the CYBE in $\frak g$ if  
$$
[r_{12},r_{13}] + [r_{12},r_{23}] + [r_{13},r_{23}] = 0 \quad \text{in } U({\frak g}),
$$  
where $U({\frak g})$ is the universal enveloping algebra of ${\frak g}$, and  
$$
r_{12} = \sum\limits_i a_i\otimes b_i\otimes 1, \quad r_{13} = \sum\limits_i a_i\otimes 1\otimes b_i, \quad r_{23} = \sum\limits_i 1\otimes a_i\otimes b_i.
$$  
Semonov-Tian-Shansky (\cite{Semonov-Tian-Shansky}) showed that if a nondegenerate symmetric invariant bilinear form exists on a Lie algebra $\frak g$, then any skew-symmetric solution $r$ of the CYBE satisfies the Rota-Baxter equation. Bai (\cite{Bai}) proved that a linear map $R: \frak g \rightarrow \frak g$ is a Rota-Baxter operator if and only if $r = R - R^{21}$ is a skew-symmetric solution of the CYBE in $\mathfrak{g} \ltimes_{{\rm ad}^{\ast}} \frak g^{\ast}$. In \cite{Pei}, all Rota-Baxter operators (of weight zero) on $\sl(2,\mathbb{C})$ were determined under the Cartan-Weyl basis, and corresponding solutions of the classical Yang-Baxter equation were found in the 6-dimensional Lie algebra $\sl(2,\mathbb{C})\ltimes_{\ad^*}\sl(2,\mathbb{C})^*$. As an application, this paper also presents the solution of CYBE derived from Nijenhuis operators on 2-dimensional complex pre-Lie algebras and 3-dimensional complex associative algebras.

Classifying algebras of small dimensions is a natural first step toward understanding higher-dimensional cases. In fact, 2-dimensional and 3-dimensional pre-Lie algebras have been completely classified, and some special cases in 4 dimensions have been classified (for example, 4-dimensional nilpotent pre-Lie algebras \cite{Adashev}). Therefore, ``2-dimensional'' and ``3-dimensional'' are the natural minimal dimensions for a nontrivial classification of this kind.  There are 5 classes of 2-dimensional commutative pre-Lie algebras and 6 classes of 2-dimensional non-commutative pre-Lie algebras (\cite{Burde}). The classification of 3-dimensional pre-Lie algebras falls into the following three categories: 3-dimensional simple pre-Lie algebras, pre-Lie algebras on the 3-dimensional Heisenberg Lie algebra, and 3-dimensional associative algebras (\cite{Bai-1}). However, their Nijenhuis operators remain unexplored. 

The novelty of this work lies in describing all Nijenhuis operators on 2-dimensional complex pre-Lie algebras and 3-dimensional complex associative algebras. In particular, we emphasize that this study provides a systematic derivation of induced solutions to the classical Yang-Baxter equation. The Nijenhuis operators presented are parametrically complete and list all possible Nijenhuis operators satisfying the defining equation. However, we do not perform a full orbit analysis under the automorphism group $\mathrm{Aut}(A)$. This paper aims to explicitly list Nijenhuis operators, enabling future conjugacy classification.

\subsection{Outline of the paper}

This paper is organized as follows. Section 2 reviews basic concepts and known results on pre-Lie algebras and Nijenhuis operators. Section 3 presents all Nijenhuis operators on 2-dimensional complex pre-Lie algebras (Theorems \ref{Ni ope on 2 c} and \ref{Ni ope on 2 non c}). Section 4 lists all Nijenhuis operators on 3-dimensional complex associative algebras (Theorem \ref{Ni ope on 3 c a } and \ref{Ni ope on 3 non c a }). In Section 5, we first outline the general construction pipeline that transforms Nijenhuis operators on pre-Lie algebras into Rota-Baxter operators on their sub-adjacent Lie algebras, and subsequently into solutions of the classical Yang-Baxter equation. We then provide a fully worked, step-by-step example for the 2-dimensional pre-Lie algebra \(B_1\) to illustrate this procedure concretely. Section 5.3 lists the sub-adjacent Lie algebras that arise from all classified 2-dimensional pre-Lie and 3-dimensional associative algebras. Section 5.4 presents the Rota-Baxter operators induced by the Nijenhuis operators on these Lie algebras. Finally, Section 5.5 gives the explicit solutions of the CYBE on the corresponding Lie algebras \({\frak g}(A) \ltimes_{{\rm ad}^{\ast}} {\frak g}(A)^{\ast}\).

All pre-Lie and associative algebras discussed in this paper are finite-dimensional over $\mathbb{C}$, and juxtaposition $kx$ represents the product of the scalar $k$ and the element $x$ in the vector space, while the dotted notation $x\cdot y$ indicates the product in a pre-Lie algebra. Denote the sub-adjacent Lie algebra ${\frak g}(B_1)$, ${\frak g}(B_2)$, ${\frak g}(B_3)$, ${\frak g}(B_4)$, and ${\frak g}(B_5)$ by ${\frak g_1}$.  Denote ${\frak g}(D_1)$, ${\frak g}(D_2)$ and ${\frak g}(D_3)$ by ${\frak g_2}$. Denote ${\frak g}(D_4)$, ${\frak g}(D_6)$, ${\frak g}(D_8)$, and ${\frak g}(D_9)$ by ${\frak g_3}$. Denote ${\frak g}(D_5)$, ${\frak g}(D_7)$ by ${\frak g_4}$.
\newpage

\section{Preliminaries}

In this section, we provide some preliminaries and classification results of 2-dimensional pre-Lie algebras  and 3-dimensional associative algebras in \cite{Burde, Bai-Meng} and \cite{Bai-1} .

\begin{definition} \rm{(\cite{Cayley})}
 Let $A$ be a vector space over
$\mathbb{C}$ with a bilinear product $\cdot:A\times A\rightarrow A$. Then $(A,\cdot)$ is called
a {\bf pre-Lie algebra} if for any $x,y,z\in A$,
\begin{equation}
(x\cdot y)\cdot z-x\cdot (y\cdot z)=(y\cdot x)\cdot z-y\cdot (x\cdot z).
\end{equation}
For a pre-Lie algebra $(A,\cdot)$, the commutator
\begin{equation}\label{commutator}
[x,y]=x\cdot y-y\cdot x,
\end{equation}
defines a Lie algebra ${\frak g}(A)$,~which is called the
{\bf sub-adjacent Lie algebra} of $(A,\cdot)$.
\end{definition}

\begin{proposition}\label{commutative pre-Lie algebra}\rm{(\cite{Burde, Bai-Meng})}
 Let $\left(A, \cdot\right)$ be a 2-dimensional commutative pre-Lie algebra and $\left\{e_{1}, e_{2}\right\}$ be a basis of $A$. Then $(A,\cdot)$ isomorphic to  one of the
following algebras:

$(A_1)$ $e_1\cdot e_1=e_1,~e_2\cdot e_2=e_2$;

$(A_2)$ $e_1\cdot e_1=e_1,~e_1.e_2=e_2\cdot e_1=e_2$;

$(A_3)$ $e_1\cdot e_1=e_1$;

$(A_4)$ $e_i\cdot e_j=0,~i,j=1,2$;

$(A_5)$ $e_1\cdot e_1=e_2$.
\end{proposition}

\begin{proposition}\label{NON commutative pre-Lie algebra}\rm{(\cite{Burde, Bai-Meng})}
Let $\left(B, \cdot\right)$ be a 2-dimensional non-commutative pre-Lie algebra and $\left\{e_{1}, e_{2}\right\}$ be a basis of $B$. Then $(B,\cdot)$ isomorphic to  one of the
following algebras:

$(B_1)$ $e_2\cdot e_1=-e_1,~e_2\cdot e_2=e_1-e_2$;

$(B_2)$ $e_2\cdot e_1=-e_1,~e_2\cdot e_2=-e_2$;

$(B_3)$ $e_2\cdot e_1=-e_1,~e_2\cdot e_2=ke_2,~k\neq-1, k\in \mathbb{C}$;

$(B_4)$ $e_1\cdot e_2=e_1,~e_2\cdot e_2=e_2$;

$(B_5)$ $e_1\cdot e_2=le_1,~e_2\cdot e_1=(l-1)e_1,~e_2\cdot e_2=e_1+le_2,~l\neq0, l\in \mathbb{C}$;

$(B_6)$ $e_1\cdot e_1=2e_1,~e_1\cdot e_2=e_2,~e_2\cdot e_2=e_1$.
\end{proposition}

\begin{proposition}\label{commutative 3 dim pre-Lie algebra}\rm{(\cite{Bai-1})}
Let $\left(C, \cdot\right)$ be a 3-dimensional commutative associative algebra and $\left\{e_{1}, e_{2}, e_{3}\right\}$ be a basis of $C$. Then $(C,\cdot)$ isomorphic to  one of the
following algebras:

$(C_1)$ $e_3\cdot e_3=e_1$;

$(C_2)$ $e_i\cdot e_j=0,~i,j=1,2,3$;

$(C_3)$ $e_2\cdot e_2=e_1,~e_3\cdot e_3=e_1$;

$(C_4)$ $e_2\cdot e_3=e_3\cdot e_2=e_1,~e_3\cdot e_3=e_2$;

$(C_5)$ $e_1\cdot e_1=e_1,~e_2\cdot e_2=e_2,~e_3\cdot e_3=e_3$;

$(C_6)$ $e_2\cdot e_2=e_2,~e_3\cdot e_3=e_3$;

$(C_7)$ $e_1\cdot e_3=e_3\cdot e_1=e_1,~e_2\cdot e_2=e_2,~e_3\cdot e_3=e_3$;

$(C_8)$ $e_3\cdot e_3=e_3$;

$(C_9)$ $e_1\cdot e_3=e_3\cdot e_1=e_1,~e_3\cdot e_3=e_3$;

$(C_{10})$ $e_1\cdot e_3=e_3\cdot e_1=e_1,~e_2\cdot e_3=e_3\cdot e_2=e_2,~e_3\cdot e_3=e_3$;

$(C_{11})$ $e_1\cdot e_1=e_2,~e_3\cdot e_3=e_3$;

$(C_{12})$ $e_1\cdot e_1=e_2,~e_1\cdot e_3=e_3\cdot e_1=e_1,~e_2\cdot e_3=e_3\cdot e_2=e_2,~e_3\cdot e_3=e_3$.
\end{proposition}

\begin{proposition}\label{non-commutative 3 dim pre-Lie algebra}\rm{(\cite{Bai-1})}
Let $\left(D, \cdot\right)$ be a 3-dimensional non-commutative associative algebra and $\left\{e_{1}, e_{2}, e_{3}\right\}$ be a basis of $D$. Then $(D,\cdot)$ isomorphic to  one of the
following algebras:

$(D_1)$ $e_1 \cdot e_2 = \frac{1}{2} e_3,~
e_2 \cdot e_1 = -\frac{1}{2} e_3$;

$(D_2)$ $e_2 \cdot e_1 = - e_3$;

$(D_3)$ $e_1 \cdot
e_1 =  e_3,~e_1 \cdot e_2 = e_3,~ e_2 \cdot e_2 =\lambda e_3,~\lambda \neq 0, \lambda\in \mathbb{C}$;

$(D_4)$ $e_3 \cdot e_2 = e_2,~e_3 \cdot e_3 =e_3$;

$(D_5)$ $e_2 \cdot e_3 =e_2,~e_3
\cdot e_3 =e_3$;

$(D_6)$ $e_1
\cdot e_1 =e_1,~e_3 \cdot e_2 =e_2,~e_3 \cdot e_3 =e_3$;

$(D_7)$ $e_1 \cdot e_1 =e_1,~e_2 \cdot e_3
=e_2,~e_3 \cdot e_3 =e_3$;

$(D_8)$ $e_1
\cdot e_3 =e_1,~e_3 \cdot e_1 =e_1,~e_3 \cdot e_2 =e_2,~e_3
\cdot e_3 =e_3$;

$(D_9)$ $e_1 \cdot e_1 =e_1,~e_1
\cdot e_2 =e_2 \cdot e_1 =e_2,~e_1 \cdot e_3 =e_3 \cdot e_1=e_3,~
e_3 \cdot e_2 =e_2,~e_3 \cdot
e_3 =e_3$;

$(D_{10})$ $e_3 \cdot e_1 =e_1,~e_3 \cdot e_2
=e_2,~e_3 \cdot e_3 =e_3$;

$(D_{11})$ $e_1 \cdot e_3 =e_1,~e_2 \cdot e_3
=e_2,~e_3 \cdot e_3 =e_3$;

$(D_{12})$ $e_3 \cdot e_1 =e_1,~e_2 \cdot e_3
=e_2,~e_3 \cdot e_3 =e_3$.
\end{proposition}

\begin{definition} \rm{\cite{Wang}}
Let $(A, \cdot)$ be a pre-Lie algebra. A linear map  $N: A \rightarrow A$~is called a {\bf Nijenhuis operator} on $(A, \cdot)$ if
\begin{equation}\label{eq:Ni ope1}
N(x) \cdot N(y)=N(N(x) \cdot  y+x \cdot N(y)-N(x \cdot  y)),\quad\forall~ x, y \in A.
\end{equation}
\end{definition}

\begin{proposition}\label{PRO: Structural constants}
Let $(A, \cdot)$  be an $n$-dimensional pre-Lie algebra and  $\left\{e_{1}, \cdots, e_{n}\right\}$  be a basis of  $A$. For all positive integers~$1\leq i, j, t\leq n$ and structural constants $C_{i j}^{t}\in \mathbb{C}$,  set
\begin{equation}\label{eq:pre-Lie-muilt}
  e_{i} \cdot e_{j}=\sum_{t=1}^{n} C_{i j}^{t} e_{t}.
\end{equation}
If  $N:A\rightarrow A$ is a linear map given by
\begin{equation}\label{eq:Nii-mulit}
  N\left(e_{i}\right)=\sum_{j=1}^{n} n_{i j} e_{j}, \quad ~n_{i j}\in\mathbb{C},
\end{equation}
then $N$ is a Nijenhuis operator  on $(A, \cdot)$ if and only if $n_{i j}$ satisfies the following equations:
\begin{equation}\label{structural constants}
  \sum_{k, l, m,t=1}^{n}\left(C_{k l}^{m} n_{i k} n_{j l}+C_{i j}^{t} n_{t l}n_{l m}-C_{k j}^{l} n_{i k} n_{l m}-C_{i l}^{k} n_{j l} n_{k m}\right)=0.
\end{equation}
\end{proposition}
\begin{proof}
 For all $e_{i}, e_{j} \in \left\{e_{1}, \cdots, e_{n}\right\}$, $1\leq i, j\leq n$, set
 \begin{equation*}
N\left(e_{i}\right)=\sum_{k=1}^{n} n_{i k} e_{k},~N\left(e_{j}\right)=\sum_{l=1}^{n} n_{j l} e_{l},~ e_{i} \cdot e_{j}=\sum_{t=1}^{n} C_{i j}^{t} e_{t},~~~~ n_{i k},n_{j l},C_{i j}^{t}\in \mathbb{C}.
 \end{equation*}
Then we have
 \begin{equation*}
  N(e_{i}) \cdot N(e_{j})=\sum_{k,l=1}^{n} n_{i k}n_{j l} e_{k}\cdot e_{l}=\sum_{k,l,m=1}^{n}C_{k l}^{m} n_{i k}n_{j l} e_{m},
 \end{equation*}
and
\begin{eqnarray*}
&&N(N(e_{i}) \cdot  e_{j}+e_{i} \cdot N(e_{j})-N(e_{i} \cdot  e_{j}))\\
&=&N(\sum_{k=1}^{n} n_{i k} e_{k}\cdot e_{j}+e_{i} \cdot \sum_{l=1}^{n} n_{j l} e_{l}-N(\sum_{t=1}^{n} C_{i j}^{t} e_{t}))\\
&=&N(\sum_{k,l=1}^{n}C_{k j}^{l} n_{i k} e_{l}+ \sum_{l,k=1}^{n}C_{i l}^{k} n_{j l} e_{k}-\sum_{t,l=1}^{n} C_{i j}^{t}  n_{t l}e_{l})\\
&=&\sum_{k,l,m,t=1}^{n}(C_{k j}^{l} n_{i k}n_{l m}e_{m}+C_{i l}^{k} n_{j l}n_{k m} e_{m}- C_{i j}^{t}  n_{t l}n_{l m}e_{m}).\\
\end{eqnarray*}

If $N$ is a Nijenhuis operator, then
\begin{eqnarray*}
   && N(e_{i}) \cdot N(e_{j})-N(N(e_{i}) \cdot  e_{j}+e_{i} \cdot N(e_{j})-N(e_{i} \cdot  e_{j})) \\
  &=& \sum_{k,l,m=1}^{n}C_{k l}^{m} n_{i k}n_{j l} e_{m}-\sum_{k,l,m,t=1}^{n}(C_{k j}^{l} n_{i k}n_{l m}e_{m}+C_{i l}^{k} n_{j l}n_{k m} e_{m}- C_{i j}^{t}  n_{t l}n_{l m}e_{m}) \\
  &= & \sum_{k,l,m,t=1}^{n}(C_{k l}^{m} n_{i k}n_{j l} e_{m}+C_{i j}^{t}  n_{t l}n_{l m}e_{m}-C_{k j}^{l} n_{i k}n_{l m}e_{m}-C_{i l}^{k} n_{j l}n_{k m} e_{m}) \\
  &= & 0.
\end{eqnarray*}
Therefore,  $N$ is a Nijenhuis operator  on $(A, \cdot)$ if and only if \eqref{structural constants} holds.
\end{proof}

\section{Nijenhuis operators on 2-dimensional pre-Lie algebras}

Based on the classification results of 2-dimensional commutative and non-commutative complex pre-Lie algebras in Propositions \ref{commutative pre-Lie algebra} and \ref{NON commutative pre-Lie algebra}, this section determines all Nijenhuis operators on $(A_i,\cdot)$ and $(B_i,\cdot)$, respectively.

\subsection{Nijenhuis operators on 2-dimensional commutative pre-Lie algebras} In this subsection, we determine all Nijenhuis operators on $(A_i,\cdot)$.

\textbf{Reading guide for Theorem \ref{Ni ope on 2 c}:} we use $N_{A_i}^j$ to denote the $j$-th Nijenhuis operator on the pre-Lie algebra $(A_i, \cdot)$,  where the algebras $(A_i, \cdot)$ are taken from the complete classification in \cite{Bai-Meng}. The parameters $n_{ij}$, where $i,j=1,2$, of the matrix representing $N$ are arbitrary complex numbers unless otherwise specified.

\begin{theorem}\label{Ni ope on 2 c}
Let $N:A_i\rightarrow A_i,1\leq i\leq 5$, be a linear map on the pre-Lie algebra $(A_i, \cdot)$ defined by \eqref{eq:Nii-mulit}. Then the Nijenhuis operators on 2-dimensional commutative pre-Lie algebras are as follows:

$(1)$~The Nijenhuis operators on the pre-Lie algebra~$(A_{1},\cdot)$~are:
\begin{itemize}
   \item[] $N_{A_{1}}^1(e_1)= n_{11}e_1,~~ N_{A_{1}}^1(e_2)= n_{21}e_1+(n_{11}+n_{21})e_2,~n_{21}\neq0.$
   \item[] $N_{A_{1}}^2(e_1)= n_{11}e_1,~~ N_{A_{1}}^2(e_2)= n_{22}e_2.$
   \item[]  $N_{A_{1}}^3(e_1)=n_{11}e_1+ n_{12}e_2 ,~~ N_{A_{1}}^3(e_2)= (n_{11}-n_{12})e_2,~n_{12}\neq0.$
   \end{itemize}

$(2)$~The Nijenhuis operators on the pre-Lie algebra $(A_{2},\cdot)$~are:
\begin{itemize}
   \item[] $N_{A_{2}}^1(e_1)= n_{11}e_1,~~ N_{A_{2}}^1(e_2)= n_{22}e_2.$
   \item[]  $N_{A_{2}}^2(e_1)=n_{11}e_1+ n_{12}e_2 ,~~ N_{A_{2}}^2(e_2)= n_{11}e_2,~n_{12}\neq0.$
   \end{itemize}

$(3)$~The Nijenhuis operators on the pre-Lie algebra~$(A_{3},\cdot)$~are:
\begin{itemize}
   \item[] $N_{A_{3}}^1(e_1)= n_{11}e_1,~~ N_{A_{3}}^1(e_2)= n_{22}e_2.$
   \item[]  $N_{A_{3}}^2(e_1)=n_{11}e_1+ n_{12}e_2 ,~~ N_{A_{3}}^2(e_2)= n_{11}e_2,~n_{12}\neq0.$
   \end{itemize}

$(4)$~The Nijenhuis operators on the pre-Lie algebra~$(A_{4},\cdot)$~are:
\begin{itemize}
   \item[] $N_{A_{4}}^1(e_1)= n_{11}e_1+n_{12}e_2,~~ N_{A_{4}}^1(e_2)= n_{21}e_1+n_{22}e_2.$
   \end{itemize}

$(5)$~The Nijenhuis operators on the pre-Lie algebra~$(A_{5},\cdot)$~are:
\begin{itemize}
   \item[]  $N_{A_{5}}^1(e_1)=n_{11}e_1+ n_{12}e_2 ,~~ N_{A_{5}}^1(e_2)= n_{11}e_2.$
   \end{itemize}
\end{theorem}
   \begin{proof}
Let $\left\{e_{1},  e_{2}\right\}$ be a basis of 2-dimensional  commutative pre-Lie algebras~$(A_i,\cdot)$.~Set
  \begin{equation*}\label{eq:An operator on a basis}
    N(e_{1})=n_{11}e_1+n_{12}e_2,~N(e_{2})=n_{21}e_1+n_{22}e_2.
  \end{equation*}
Since $N$ is a Nijenhuis operator,~we have
 \begin{eqnarray}
    N(e_{1}) \cdot N(e_{1})&=&N(N(e_{1}) \cdot  e_{1}+e_{1} \cdot N(e_{1})-N(e_{1} \cdot  e_{1})),\label{eq:base e1e1}\\
     N(e_{1}) \cdot N(e_{2})&=&N(N(e_{1}) \cdot  e_{2}+e_{1} \cdot N(e_{2})-N(e_{1} \cdot  e_{2})),\label{eq:base e1e2}\\
      N(e_{2}) \cdot N(e_{1})&=&N(N(e_{2}) \cdot  e_{1}+e_{2} \cdot N(e_{1})-N(e_{2} \cdot  e_{1})),\label{eq:base e2e1}\\
       N(e_{2}) \cdot N(e_{2})&=&N(N(e_{2}) \cdot  e_{2}+e_{2} \cdot N(e_{2})-N(e_{2} \cdot  e_{2})).\label{eq:base e2e2}
 \end{eqnarray}
Calculating~\eqref{eq:base e1e1}~for the pre-Lie algebra~$(A_{1},\cdot)$, we have:
\begin{equation*}
   N(e_{1}) \cdot N(e_{1})
   = {n_{11}}^2e_1+{n_{12}}^2e_2,\label{eq11}
\end{equation*}
and
\begin{equation*}
  N(N(e_{1}) \cdot  e_{1}+e_{1} \cdot N(e_{1})-N(e_{1} \cdot  e_{1}))
  = (n_{11}^2-n_{12}n_{21})e_1+(n_{11}n_{12}-n_{12}n_{22})e_2.\label{eq12}
\end{equation*}
Then
\begin{equation}\label{A11}
  n_{{11}}^2=n_{11}^2-n_{12}n_{21},~{n_{12}}^2=n_{11}n_{12}-n_{12}n_{22}.
\end{equation}
Similarly,~by (\ref{eq:base e1e2}) $\sim$ (\ref{eq:base e2e2}), we have:
\begin{eqnarray}
  n_{11}n_{21}=n_{11}n_{21}+n_{12}n_{21},~n_{12}n_{22}=n_{12}n_{21}+n_{12}n_{22}\label{A12}. \\
  {n_{21}}^2 = n_{21}n_{22}-n_{11}n_{21},~{n_{22}}^2={n_{22}}^2-n_{12}n_{21}\label{A13}.
\end{eqnarray}
Simplifying~\eqref{A11}$\sim$\eqref{A13}, we obtain the following three equations:
\begin{equation}\label{A14}
\left\{\begin{array}{l}
n_{12}n_{21}=0,\\{n_{12}}^2=n_{11}n_{12}-n_{12}n_{22},\\
{n_{21}}^2=n_{21}n_{22}-n_{11}n_{21}.
\end{array}\right.
\end{equation}

To solve the quadratic equations, we distinguish the two cases depending on whether
or not $n_{12} = 0$.

$\mathbf{Case~1}$: If $n_{12}=0$,~then~\eqref{A14} implies~$n_{21}(n_{21}-n_{22}+n_{11})=0$.~There are two subcases:~$n_{21}=0$ and $n_{21}\neq0$. If $n_{21}\neq0$,~we have~$n_{11}+n_{21}=n_{22}$.~Then we get ~$N_{A_{1}}^1$.
If $n_{21}=0$,~then we get $N_{A_{1}}^2$.

$\mathbf{Case~2}$: If $n_{12}\neq0$, we have $n_{21}=0,~n_{22}=n_{11}-n_{12}$.~Then we get $N_{A_{1}}^3$.
Therefore, we obtain all Nijenhuis operators on the pre-Lie algebra $(A_1,\cdot)$.

Similarly, we can obtain all Nijenhuis operators on other 2-dimensional  commutative pre-Lie algebras. This completes the proof.
\end{proof}

\subsection{Nijenhuis operators on 2-dimensional non-commutative pre-Lie algebras}
In this subsection, we determine all Nijenhuis operators on $(B_i,\cdot)$.

\textbf{Reading guide for Theorem \ref{Ni ope on 2 non c}:} we use $N_{B_{i}}^j$ to denote the $j$-th Nijenhuis operator on the pre-Lie algebra $(B_{i}, \cdot)$, where the algebras $(B_{i}, \cdot)$ are taken from the complete classification in \cite{Bai-Meng}. The parameters $n_{ij}$, where $i,j=1,2$, of the matrix representing $N$ are arbitrary complex numbers unless otherwise specified. 
The following results are obtained through direct computation.
\begin{theorem}\label{Ni ope on 2 non c}
Let $N:B_i\rightarrow B_i$, $1\leq i\leq 6$, be the linear map on the pre-Lie algebras $(B_{i},\cdot)$ defined by \eqref{eq:Nii-mulit}. Then the Nijenhuis operators on 2-dimensional non-commutative pre-Lie algebras are as follows:

$(1)$~The Nijenhuis operators on the pre-Lie algebra~$(B_{1},\cdot)$~are:
\begin{itemize}
   \item[]  $N_{B_{1}}^1(e_1)=n_{11}e_1,~~ N_{B_{1}}^1(e_2)= n_{21}e_1+ n_{11}e_2.$
   \end{itemize}

$(2)$~The Nijenhuis operators on the pre-Lie algebra~$(B_{2},\cdot)$~are:
\begin{itemize}
   \item[]  $N_{B_{2}}^1(e_1)=n_{11}e_1+ n_{12}e_2 ,~~ N_{B_{2}}^1(e_2)= n_{21}e_1+ n_{22}e_2.$
   \end{itemize}

$(3)$~The Nijenhuis operators on the pre-Lie algebra~$(B_{3},\cdot)$~are:
\begin{itemize}
   \item[]  $N_{B_{3}}^1(e_1)=n_{11}e_1 ,~~ N_{B_{3}}^1(e_2)= n_{21}e_1+ n_{11}e_2,~n_{21}\neq0.$
   \item[]  $N_{B_{3}}^2(e_1)=n_{11}e_1 ,~~ N_{B_{3}}^2(e_2)= n_{22}e_2.$
 \end{itemize}

$(4)$~The Nijenhuis operators on the pre-Lie algebra~$(B_{4},\cdot)$~are:
\begin{itemize}
   \item[]  $N_{B_{4}}^1(e_1)=n_{11}e_1+ n_{12}e_2 ,~~ N_{B_{4}}^1(e_2)= n_{21}e_1+ n_{22}e_2.$
   \end{itemize}

$(5)$~The Nijenhuis operators on the pre-Lie algebra~$(B_{5},\cdot)$~are:
\begin{itemize}
   \item[]  $N_{B_{5}}^1(e_1)=n_{11}e_1 ,~~ N_{B_{5}}^1(e_2)= n_{21}e_1+ n_{11}e_2.$
   \end{itemize}
   \begin{itemize}
   \item[]  $N_{B_{5}}^2(e_1)=n_{11}e_1 ,~~ N_{B_{5}}^2(e_2)= \frac{n_{11}-n_{22}}{l-1}e_1+ n_{22}e_2,~n_{11}\neq n_{22},~l\neq1.$
   \end{itemize}

$(6)$~The Nijenhuis operators on the pre-Lie algebra~$(B_{6},\cdot)$~are:
\begin{itemize}
   \item[]  $N_{B_{6}}^1(e_1)=n_{11}e_1 ,~~ N_{B_{6}}^1(e_2)= n_{21}e_1+ (n_{11}\pm n_{21}\sqrt {-1})e_2.$
   \end{itemize}
\begin{itemize}
   \item[]  $N_{B_{6}}^2(e_1)=n_{11}e_1+n_{12}e_2 ,~~ N_{B_{6}}^2(e_2)= -n_{12}e_1+n_{11}e_2,~n_{12}\neq0.$
   \end{itemize}
\end{theorem}

\begin{remark}
Theorem \ref{Ni ope on 2 c} and Theorem \ref{Ni ope on 2 non c} encompass all possible Nijenhuis operators on 2-dimensional pre-Lie algebras, and the above results are summarized in Table \ref{tab:prelie_algebras_nijenhuis} of the Appendix. 
\end{remark}

\section{Nijenhuis operators on 3-dimensional associative algebras}

Associative algebras are a special type of pre-Lie algebras. In this section, we use the classification results of commutative and non-commutative associative algebras from Proposition \ref{commutative 3 dim pre-Lie algebra} and \ref{non-commutative 3 dim pre-Lie algebra} to find all Nijenhuis operators on $(C_{i},\cdot)$ and $(D_{i},\cdot)$, respectively.

\subsection{Nijenhuis operators on 3-dimensional commutative associative algebras}
In this subsection, we determine all Nijenhuis operators on $(C_i,\cdot)$.

\textbf{Reading guide for Theorem \ref{Ni ope on 3 c a }:} we use $N_{C_{i}}^j$ to denote the $j$-th Nijenhuis operator on the commutative associative algebra $(C_{i}, \cdot)$, where the algebra $(C_{i}, \cdot)$ are taken from the complete classification in \cite{Bai-1}. The parameters $n_{ij}$, where $i,j=1,2,3$, of the matrix representing $N$ are arbitrary complex numbers unless otherwise specified. 

\begin{theorem}\label{Ni ope on 3 c a }
Let $N:C_i\rightarrow C_i$, $1\leq i \leq 12$, be a linear map on commutative associative algebras $(C_{i},\cdot)$ defined by \eqref{eq:Nii-mulit}. Then the Nijenhuis operators on 3-dimensional commutative associative algebras are as follows:

$(1)$~The Nijenhuis operators on the associative algebra~$(C_{1},\cdot)$~are:
\begin{itemize}
   \item[]  $N_{C_{1}}^1(e_1)=n_{11}e_1 ,~~ N_{C_{1}}^1(e_2)= n_{21}e_1+ n_{22}e_2,~~ N_{C_{1}}^1(e_3)= n_{31}e_1+ n_{32}e_2+ n_{11}e_3.$
 \item[]  $N_{C_{1}}^2(e_1)=n_{11}e_1+ n_{12}e_2 ,~~ N_{C_{1}}^2(e_2)=-\frac{(n_{11}-n_{33})^2}{n_{12}}e_1 + (2n_{33}-n_{11}) e_2,~~ N_{C_{1}}^2(e_3)= n_{31}e_1+ n_{32}e_2+ n_{33}e_3,~n_{12}\neq0.$
   \end{itemize}

$(2)$~The Nijenhuis operators on the associative algebra~$(C_{2},\cdot)$~are:
\begin{itemize}
   \item[]  $N_{C_{2}}^1(e_1)=n_{11}e_1+ n_{12}e_2+ n_{13}e_3 ,~~ N_{C_{2}}^1(e_2)= n_{21}e_1+ n_{22}e_2+ n_{23}e_3 ,~~ N_{C_{2}}^1(e_3)= n_{31}e_1+ n_{32}e_2+ n_{33}e_3.$
   \end{itemize}

$(3)$~The Nijenhuis operators on the associative algebra~$(C_{3},\cdot)$~are:
\begin{itemize}
    \item[]  $N_{C_{3}}^1(e_1)=n_{11}e_1 ,
   ~~ N_{C_{3}}^1(e_2)= n_{21}e_1+ n_{11}e_2,~~
   N_{C_{3}}^1(e_3)= n_{31}e_1+ n_{11}e_3.$
   \item[]  $N_{C_{3}}^2(e_1)=n_{11}e_1,
   ~~ N_{C_{3}}^2(e_2)= n_{21}e_1+ n_{11}e_2,
  N_{C_{3}}^2(e_3)= n_{31}e_1+ n_{32}e_2+ (n_{11}\pm n_{32}\sqrt{-1})e_3,~~n_{32}\neq0.$
   \item[]  $N_{C_{3}}^3(e_1)=n_{11}e_1,
   ~~ N_{C_{3}}^3(e_2)= n_{21}e_1+ (n_{11}\pm n_{23}\sqrt{-1})e_2+ n_{23}e_3,~
 N_{C_{3}}^3(e_3)= n_{31}e_1+ n_{11}e_3,~n_{23}\neq0.$
     \item[]  $N_{C_{3}}^4(e_1)=n_{11}e_1 ,
   ~~ N_{C_{3}}^4(e_2)= n_{21}e_1+ (n_{11}-n_{23}\sqrt{-1})e_2+n_{23}e_3,~~
   N_{C_{3}}^4(e_3)= n_{31}e_1+ n_{32}e_2+ (n_{11}+ n_{32}\sqrt{-1})e_3,~n_{23},~n_{32}\neq0.$
    \item[]  $N_{C_{3}}^5(e_1)=n_{11}e_1 ,
   ~~ N_{C_{3}}^5(e_2)= n_{21}e_1+ (n_{11}+n_{23}\sqrt{-1})e_2+n_{23}e_3,~~
   N_{C_{3}}^5(e_3)= n_{31}e_1+ n_{32}e_2+ (n_{11}- n_{32}\sqrt{-1})e_3,~n_{23},~n_{32}\neq0.$
   \end{itemize}

$(4)$~The Nijenhuis operators on the associative algebra~$(C_{4},\cdot)$~are:
\begin{itemize}
    \item[]  $N_{C_{4}}^1(e_1)=n_{11}e_1 ,
   ~~ N_{C_{4}}^1(e_2)= n_{21}e_1+ n_{11}e_2,~~
   N_{C_{4}}^1(e_3)= n_{31}e_1+ n_{32}e_2+ n_{11}e_3.$
   \end{itemize}

$(5)$~The Nijenhuis operators on the associative algebra~$(C_{5},\cdot)$~are:
\begin{itemize}
    \item[]  $N_{C_{5}}^1(e_1)=n_{11}e_1 ,
   ~~ N_{C_{5}}^1(e_2)=  n_{22}e_2,~~
   N_{C_{5}}^1(e_3)=  n_{33}e_3.$
   \item[]  $N_{C_{5}}^2(e_1)=n_{11}e_1 ,
   ~~ N_{C_{5}}^2(e_2)=  n_{22}e_2,~~
   N_{C_{5}}^2(e_3)=n_{31}e_1+  (n_{11}+n_{31})e_3,~n_{31}\neq0.$
    \item[]  $N_{C_{5}}^3(e_1)=n_{11}e_1 ,
   ~~ N_{C_{5}}^3(e_2)=  n_{22}e_2,~~
   N_{C_{5}}^3(e_3)=n_{32}e_2+  (n_{22}+n_{32})e_3,~n_{32}\neq0.$
    \item[]  $N_{C_{5}}^4(e_1)=n_{11}e_1 ,
   ~~ N_{C_{5}}^4(e_2)=  n_{21}e_1+  (n_{11}+n_{21})e_2,~~
   N_{C_{5}}^4(e_3)=n_{33}e_3,~n_{21}\neq0.$
    \item[]  $N_{C_{5}}^5(e_1)=n_{11}e_1 ,
   ~~ N_{C_{5}}^5(e_2)=(n_{23}+n_{33})e_2+  n_{23}e_3,~~
   N_{C_{5}}^5(e_3)=n_{33}e_3,~n_{23}\neq0.$
   \item[]  $N_{C_{5}}^6(e_1)=n_{11}e_1 ,
   ~~ N_{C_{5}}^6(e_2)=  n_{21}e_1+  (n_{11}+n_{21})e_2,~~
   N_{C_{5}}^6(e_3)=n_{32}e_1+n_{32}e_2+(n_{11}+n_{21}+n_{32})e_3,
   ~n_{21},~n_{32}\neq0.$
    \item[]  $N_{C_{5}}^7(e_1)=(n_{22}+n_{32}-n_{31})e_1 ,
   ~~ N_{C_{5}}^7(e_2)= n_{22}e_2,~~
   N_{C_{5}}^7(e_3)=n_{31}e_1+n_{32}e_2+(n_{22}+n_{32})e_3,
   ~n_{31},~n_{32}\neq0.$
   \item[]  $N_{C_{5}}^8(e_1)=n_{11}e_1+n_{13}e_3 ,
   ~~ N_{C_{5}}^8(e_2)=  n_{22}e_2,~~
   N_{C_{5}}^8(e_3)= (n_{11}+n_{13})e_3,~n_{13}\neq0.$
     \item[]  $N_{C_{5}}^9(e_1)=n_{11}e_1+n_{12}e_2 ,
   ~~ N_{C_{5}}^9(e_2)= (n_{11}-n_{12})e_2,~~
   N_{C_{5}}^9(e_3)= n_{33}e_3,~n_{12}\neq0.$
      \item[]  $N_{C_{5}}^{10}(e_1)=n_{11}e_1+n_{13}e_3 ,
   ~~ N_{C_{5}}^{10}(e_2)=  n_{23}e_1+  (n_{11}+n_{23})e_2+n_{23}e_3,~~
   N_{C_{5}}^{10}(e_3)=(n_{11}-n_{13})e_3,
   ~n_{13},~n_{23}\neq0.$
      \item[]  $N_{C_{5}}^{11}(e_1)=n_{11}e_1+n_{12}e_2 ,
   ~~ N_{C_{5}}^{11}(e_2)=   (n_{11}-n_{12})e_2,~~
   N_{C_{5}}^{11}(e_3)=n_{31}e_1+n_{31}e_2+(n_{31}+n_{11})e_3,
   ~n_{12},~n_{31}\neq0.$
   \item[]  $N_{C_{5}}^{12}(e_1)=(n_{23}+n_{33}-n_{21})e_1 ,
   ~~ N_{C_{5}}^{12}(e_2)= n_{21}e_1+ (n_{23}+n_{33})e_2+n_{23}e_3,~~
   N_{C_{5}}^{12}(e_3)=n_{33}e_3,
   ~n_{21},~n_{23}\neq0.$
   \item[]  $N_{C_{5}}^{13}(e_1)=(n_{22}-n_{23}-n_{31})e_1 ,
   ~~ N_{C_{5}}^{13}(e_2)= n_{23}e_1+n_{22} e_2+n_{23}e_3,~~
   N_{C_{5}}^{13}(e_3)=n_{31}e_1+(n_{22}-n_{23})e_3,
   ~n_{31},~n_{23}\neq0.$
      \item[]  $N_{C_{5}}^{14}(e_1)=n_{11}e_1+ n_{12}e_2+ n_{13}e_3 ,
   ~~ N_{C_{5}}^{14}(e_2)= (n_{11}-n_{12})e_2,~~
   N_{C_{5}}^{14}(e_3)=(n_{11}-n_{13})e_3,
   ~n_{12},~n_{13}\neq0.$
   \item[]  $N_{C_{5}}^{15}(e_1)=n_{11}e_1+ n_{12}e_2+ n_{12}e_3 ,
   ~~ N_{C_{5}}^{15}(e_2)= (n_{11}-n_{12}-n_{32}) e_2,~~
   N_{C_{5}}^{15}(e_3)=n_{32}e_2+(n_{11}-n_{12})e_3,
   ~n_{12},~n_{13},~n_{32}\neq0.$
     \item[]  $N_{C_{5}}^{16}(e_1)=n_{11}e_1+ n_{12}e_2+ n_{12}e_3 ,
   ~~ N_{C_{5}}^{16}(e_2)= (n_{11}-n_{12}) e_2+ n_{23}e_3,~~
   N_{C_{5}}^{16}(e_3)=(n_{11}-n_{12}-n_{23})e_3,
   ~n_{12},~n_{13},~n_{23}\neq0.$
   \end{itemize}

$(6)$~The Nijenhuis operators on the associative algebra~$(C_{6},\cdot)$~are:
\begin{itemize}
    \item[]  $N_{C_{6}}^1(e_1)=n_{11}e_1 ,
   ~~ N_{C_{6}}^1(e_2)=  n_{22}e_2,~~
   N_{C_{6}}^1(e_3)=  n_{33}e_3.$
   \item[]  $N_{C_{6}}^2(e_1)=n_{11}e_1 ,
   ~~ N_{C_{6}}^2(e_2)=  n_{22}e_2,~~
   N_{C_{6}}^2(e_3)=n_{31}e_1+  n_{11}e_3,~n_{31}\neq0.$
    \item[]  $N_{C_{6}}^3(e_1)=n_{11}e_1 ,
   ~~ N_{C_{6}}^3(e_2)=n_{21}e_1+  n_{11}e_2,~~
   N_{C_{6}}^3(e_3)=n_{33}e_3,~n_{21}\neq0.$
    \item[]  $N_{C_{6}}^4(e_1)=n_{11}e_1 ,
   ~~ N_{C_{6}}^4(e_2)=  n_{21}e_1+  n_{11}e_2,~~
   N_{C_{6}}^4(e_3)=n_{31}e_1+n_{11}e_3,~n_{21},~n_{31}\neq0.$
    \item[]  $N_{C_{6}}^5(e_1)=n_{11}e_1 ,
   ~~ N_{C_{6}}^5(e_2)=n_{22}e_2,~~
   N_{C_{6}}^5(e_3)=n_{32}e_2+(n_{22}+n_{32})e_3,~n_{32}\neq0.$
   \item[]  $N_{C_{6}}^6(e_1)=n_{11}e_1 ,
   ~~ N_{C_{6}}^6(e_2)=   (n_{23}+n_{33})e_2+n_{23}e_3,~~
   N_{C_{6}}^6(e_3)=n_{33}e_3,
   ~n_{23}\neq0.$
    \item[]  $N_{C_{6}}^7(e_1)=(n_{22}+n_{32})e_1 ,
   ~~ N_{C_{6}}^7(e_2)=   n_{22}e_2,~~
   N_{C_{6}}^7(e_3)=n_{31}e_1+n_{32}e_2+(n_{22}+n_{32})e_3,
   ~n_{31},~n_{32}\neq0.$
    \item[]  $N_{C_{6}}^8(e_1)=(n_{23}+n_{33})e_1 ,
    N_{C_{6}}^8(e_2)= n_{21}e_1+  (n_{23}+n_{33})e_2+n_{23}e_3,
   N_{C_{6}}^8(e_3)=n_{33}e_3,
   n_{21},n_{23}\neq0.$
   \end{itemize}

$(7)$~The Nijenhuis operators on the associative algebra~$(C_{7},\cdot)$~are:
\begin{itemize}
   \item[]  $N_{C_{7}}^1(e_1)=n_{11}e_1 ,
   ~~ N_{C_{7}}^1(e_2)=  n_{22}e_2,~~
   N_{C_{7}}^1(e_3)=  n_{33}e_3.$
   \item[]  $N_{C_{7}}^2(e_1)=n_{11}e_1 ,
   ~~ N_{C_{7}}^2(e_2)=  n_{22}e_2,~~
   N_{C_{7}}^2(e_3)=n_{31}e_1+  n_{11}e_3,~n_{31}\neq0.$
   \item[]  $N_{C_{7}}^3(e_1)=n_{11}e_1 ,
   ~~ N_{C_{7}}^3(e_2)=n_{21}e_1+  n_{11}e_2,~~
   N_{C_{7}}^3(e_3)=n_{31}e_1+n_{11}e_3,~n_{21}\neq0.$
    \item[]  $N_{C_{7}}^4(e_1)=(n_{22}+n_{32})e_1 ,
   ~~ N_{C_{7}}^4(e_2)=  n_{22}e_2,~~
   N_{C_{7}}^4(e_3)=n_{31}e_1+n_{32}e_2+(n_{22}+n_{32})e_3,~n_{31},~n_{32}\neq0.$
    \item[]  $N_{C_{7}}^5(e_1)=n_{11}e_1 ,
   ~~ N_{C_{7}}^5(e_2)=n_{22}e_2,~~
   N_{C_{7}}^5(e_3)=n_{32}e_2+(n_{22}+n_{32})e_3,~n_{32}\neq0.$
   \item[]  $N_{C_{7}}^6(e_1)=n_{11}e_1 ,
   ~~ N_{C_{7}}^6(e_2)=   (n_{23}+n_{33})e_2+n_{23}e_3,~~
   N_{C_{7}}^6(e_3)=n_{33}e_3,
   ~n_{23}\neq0.$
    \item[]  $N_{C_{7}}^7(e_1)=n_{22}e_1 ,
   ~~ N_{C_{7}}^7(e_2)= n_{21}e_1+  n_{22}e_2,~~
   N_{C_{7}}^7(e_3)=-n_{21}e_1+n_{32}e_2+(n_{22}+n_{32})e_3,
   ~n_{21},~n_{32}\neq0.$
    \item[]  $N_{C_{7}}^8(e_1)=n_{11}e_1 ,
   ~~ N_{C_{7}}^8(e_2)= n_{21}e_1+  (n_{23}+n_{11})e_2+n_{23}e_3,~~
   N_{C_{7}}^8(e_3)=-n_{21}e_1+n_{11}e_3,
   ~n_{21},~n_{23}\neq0.$
   \end{itemize}

$(8)$~The Nijenhuis operators on the associative algebra~$(C_{8},\cdot)$~are:
\begin{itemize}
  \item[]  $N_{C_{8}}^1(e_1)=n_{11}e_1 +n_{12}e_2,
   ~~ N_{C_{8}}^1(e_2)= n_{21}e_1 + n_{22}e_2,~~
   N_{C_{8}}^1(e_3)=  n_{33}e_3.$
   \item[]  $N_{C_{8}}^2(e_1)=n_{11}e_1 +n_{12}e_2 ,
   ~~ N_{C_{8}}^2(e_2)=  n_{22}e_2,~~
   N_{C_{8}}^2(e_3)=n_{32}e_2+  n_{22}e_3,~n_{32}\neq0.$
   \item[]  $N_{C_{8}}^3(e_1)= (n_{33}-\frac{n_{21}n_{32}}{n_{31}}) e_1 + \frac{n_{32}n_{33}-n_{22}n_{32}}{n_{31}}e_2,
   ~~ N_{C_{8}}^3(e_2)=n_{21}e_1+  n_{22}e_2,~~
   N_{C_{8}}^3(e_3)=n_{31}e_1+  n_{32}e_2+n_{33}e_3,~n_{31}\neq0.$
   \end{itemize}

$(9)$~The Nijenhuis operators on the associative algebra~$(C_{9},\cdot)$~are:
\begin{itemize}
   \item[]  $N_{C_{9}}^1(e_1)=n_{11}e_1 ,
   ~~ N_{C_{9}}^1(e_2)=  n_{22}e_2,~~
   N_{C_{9}}^1(e_3)=  n_{33}e_3.$
   \item[]  $N_{C_{9}}^2(e_1)=n_{11}e_1 ,
   ~~ N_{C_{9}}^2(e_2)=  n_{22}e_2,~~
   N_{C_{9}}^2(e_3)=n_{32}e_2+  n_{22}e_3,~n_{32}\neq0.$
    \item[]  $N_{C_{9}}^3(e_1)=n_{11}e_1 ,
   ~~ N_{C_{9}}^3(e_2)=  n_{22}e_2,~~
   N_{C_{9}}^3(e_3)=n_{31}e_1+  n_{11}e_3,~n_{31}\neq0.$
    \item[]  $N_{C_{9}}^4(e_1)=n_{11}e_1 ,
   ~~ N_{C_{9}}^4(e_2)= n_{11}e_2,~~
   N_{C_{9}}^4(e_3)=n_{31}e_1+ n_{32}e_2+n_{11}e_3,~n_{31},~n_{32}\neq0.$
    \item[]  $N_{C_{9}}^5(e_1)=n_{11}e_1 ,
   ~~ N_{C_{9}}^5(e_2)=n_{21}e_1+n_{22}e_2,~~
   N_{C_{9}}^5(e_3)=n_{31}e_1+n_{11}e_3,~n_{21}\neq0.$
   \item[]  $N_{C_{9}}^6(e_1)=n_{11}e_1 +n_{12}e_2,
   ~~ N_{C_{9}}^6(e_2)= n_{22}e_2,~~
   N_{C_{9}}^6(e_3)=n_{32}e_2+n_{22}e_3,
   ~n_{12}\neq0.$
   \end{itemize}

$(10)$~The Nijenhuis operators on the associative algebra~$(C_{10},\cdot)$~are:
\begin{itemize}
  \item[]  $N_{C_{10}}^1(e_1)=n_{11}e_1 +n_{12}e_2,
   ~~ N_{C_{10}}^1(e_2)= n_{21}e_1 + n_{22}e_2,~~
   N_{C_{10}}^1(e_3)=  n_{33}e_3.$
   \item[]  $N_{C_{10}}^2(e_1)=n_{11}e_1 +n_{12}e_2 ,
   ~~ N_{C_{10}}^2(e_2)=  n_{22}e_2,~~
   N_{C_{10}}^2(e_3)=n_{32}e_2+  n_{22}e_3,~n_{32}\neq0.$
   \item[]  $N_{C_{10}}^3(e_1)= (n_{33}-\frac{n_{21}n_{32}}{n_{31}}) e_1 + \frac{n_{32}n_{33}-n_{22}n_{32}}{n_{31}}e_2,
   ~~ N_{C_{10}}^3(e_2)=n_{21}e_1+  n_{22}e_2,~~
   N_{C_{10}}^3(e_3)=n_{31}e_1+  n_{32}e_2+n_{33}e_3,~n_{31}\neq0.$
   \end{itemize}

$(11)$~The Nijenhuis operators on the associative algebra~$(C_{11},\cdot)$~are:
\begin{itemize}
   \item[]  $N_{C_{11}}^1(e_1)=n_{11}e_1+n_{12}e_2 ,
   ~~ N_{C_{11}}^1(e_2)=  n_{11}e_2,~~
   N_{C_{11}}^1(e_3)=  n_{33}e_3.$
   \item[]  $N_{C_{11}}^2(e_1)=n_{11}e_1+n_{12}e_2  ,
   ~~ N_{C_{11}}^2(e_2)=  n_{11}e_2,~~
   N_{C_{11}}^2(e_3)=n_{32}e_2+  n_{11}e_3,~n_{32}\neq0.$
    \item[]  $N_{C_{11}}^3(e_1)=n_{11}e_1-n_{31}e_2 ,
   ~~ N_{C_{11}}^3(e_2)=  n_{11}e_2,~~
   N_{C_{11}}^3(e_3)=n_{31}e_1+n_{32}e_2+  n_{11}e_3,~n_{31}\neq0.$
   \end{itemize}

$(12)$~The Nijenhuis operators on the associative algebra~$(C_{12},\cdot)$~are:
\begin{itemize}
    \item[]  $N_{C_{12}}^1(e_1)=n_{11}e_1+n_{12}e_2 ,
   ~~ N_{C_{12}}^1(e_2)=  n_{11}e_2,~~
   N_{C_{12}}^1(e_3)=  n_{33}e_3.$
   \item[]  $N_{C_{12}}^2(e_1)=n_{11}e_1+n_{12}e_2  ,
   ~~ N_{C_{12}}^2(e_2)=  n_{11}e_2,~~
   N_{C_{12}}^2(e_3)=n_{32}e_2+  n_{11}e_3,~n_{32}\neq0.$
    \item[]  $N_{C_{12}}^3(e_1)=n_{11}e_1+n_{31}e_2 ,
   ~~ N_{C_{12}}^3(e_2)=  n_{11}e_2,~~
   N_{C_{12}}^3(e_3)=n_{31}e_1+n_{32}e_2+  n_{11}e_3,~n_{31}\neq0.$
   \end{itemize}
\end{theorem}
\begin{proof}
Let $\left\{e_{1},  e_{2}, e_{3}\right\}$ be a basis of 3-dimensional associative algebras~$(C_{i},\cdot)$, and let~$N$ be a linear map on~$C_{i}$. For all $n_{ij}\in
 \mathbb{C},~1\leqslant i,j\leqslant3$,~set
  \begin{equation*}
    N(e_{1})=n_{11}e_1+n_{12}e_2+n_{13}e_3,
    ~N(e_{2})=n_{21}e_1+n_{22}e_2+n_{23}e_3,
    ~N(e_{3})=n_{31}e_1+n_{32}e_2+n_{33}e_3.
  \end{equation*}
The non-zero structural constants of $(C_{1},\cdot)$ is
\begin{equation*}
   C_{3 3}^{1}=1.
\end{equation*}
By Proposition \ref{PRO: Structural constants}, $N$ is a Nijenhuis operator  on $C_{1}$ if and only if $n_{i j} $ satisfies the following equations:
\begin{equation}\label{3-dim ass alg C_1 equ}
  \left\{\begin{array}{l}
n_{13}=0,\\n_{23}=0,\\
({n_{11}}-n_{33})^2=-n_{12}n_{21},\\
n_{12}(2n_{33}-n_{11}-n_{22})=0.\\
\end{array}\right.
\end{equation}
To solve (\ref{3-dim ass alg C_1 equ}), we distinguish
the two cases depending on whether~$n_{12}$~is equal to zero or not.

$\mathbf{Case~1}$: If $n_{12} = 0$, then $n_{11} = n_{33}$, and we obtain $N_{C_{1}}^1$.

$\mathbf{Case~2}$: If $n_{12}\neq0$,~then we have $n_{21}=-\frac{(n_{11}-n_{33})^2}{n_{12}},~n_{22}=2n_{33}-n_{11}$.~Thus we get~$N_{C_{1}}^2$.

Similarly, we can obtain all Nijenhuis operators on other 3-dimensional  commutative associative algebras. This completes the proof.
\end{proof}

\subsection{Nijenhuis operators on 3-dimensional non-commutative associative algebras}
In this subsection, we determine all Nijenhuis operators on $(D_i,\cdot)$. 

\textbf{Reading guide for Theorem \ref{Ni ope on 3 non c a }}: we use $N_{D_{i}}^j$ to denote the $j$-th Nijenhuis operator on the non-commutative associative algebra $(D_i, \cdot)$, where the algebras $(D_i, \cdot)$ are taken from the complete classification in \cite{Bai-1}. The parameters $n_{ij}$, where $i,j=1,2,3$, of the matrix representing $N$ are arbitrary complex numbers unless otherwise specified.

\begin{theorem}\label{Ni ope on 3 non c a }
Let $N:D_i\rightarrow D_i$, $1 \leq i \leq 12$, be a linear map on non-commutative associative algebras $(D_i, \cdot)$ defined by \eqref{eq:Nii-mulit}. The Nijenhuis operators on 3-dimensional non-commutative associative algebras are then given as follows:

$(1)$~The Nijenhuis operators on the associative algebra~$(D_{1},\cdot)$~are:
\begin{itemize}
   \item[]  $N_{D_{1}}^1(e_1)=n_{11}e_1+  n_{13}e_3 ,~~ N_{D_{1}}^1(e_2)= n_{21}e_1+ n_{22}e_2+ n_{23}e_3 ,~~ N_{D_{1}}^1(e_3)=  n_{11}e_3.$
         \item[]  $N_{D_{1}}^2(e_1)=n_{11}e_1+n_{12}e_2+  n_{13}e_3 ,~~ N_{D_{1}}^2(e_2)= n_{22}e_2+ n_{23}e_3 ,~~ N_{D_{1}}^2(e_3)=  n_{11}e_3,~n_{12}\neq0.$
         \item[]  $N_{D_{1}}^3(e_1)=n_{11}e_1+n_{12}e_2+  n_{13}e_3 ,~~ N_{D_{1}}^3(e_2)= n_{21}e_1+ (n_{33}-\frac{n_{12}n_{21}}{n_{33}-n_{11}})e_2+ n_{23}e_3 ,~~ N_{D_{1}}^3(e_3)=  n_{33}e_3,~n_{33}\neq n_{11}.$
   \end{itemize}

$(2)$~The Nijenhuis operators on the associative algebra~$(D_{2},\cdot)$~are:
\begin{itemize}
  \item[]  $N_{D_{2}}^1(e_1)=n_{11}e_1+  n_{13}e_3 ,~~ N_{D_{2}}^1(e_2)= n_{21}e_1+ n_{22}e_2+ n_{23}e_3 ,~~ N_{D_{2}}^1(e_3)=  n_{22}e_3,~n_{21}\neq0.$
         \item[]  $N_{D_{2}}^2(e_1)=n_{11}e_1+n_{12}e_2+  n_{13}e_3 ,~~ N_{D_{2}}^2(e_2)= n_{22}e_2+ n_{23}e_3 ,~~ N_{D_{2}}^2(e_3)=  n_{11}e_3,~n_{12}\neq0.$
         \item[]  $N_{D_{2}}^3(e_1)=n_{11}e_1+  n_{13}e_3 ,~~ N_{D_{2}}^3(e_2)= n_{22}e_2+ n_{23}e_3 ,
             ~~ N_{D_{2}}^3(e_3)=  n_{22}e_3.$
          \item[]  $N_{D_{2}}^4(e_1)=n_{33}e_1+  n_{13}e_3 ,~~ N_{D_{2}}^4(e_2)= n_{22}e_2+ n_{23}e_3 ,
             ~~ N_{D_{2}}^4(e_3)=  n_{33}e_3,~n_{33}\neq n_{22}.$
   \end{itemize}

$(3)$~The Nijenhuis operators on the associative algebra~$(D_{3},\cdot)$~are:
\begin{itemize}
  \item[]  $N_{D_{3}}^1(e_1)=(n_{33}+\frac{(-1+\sqrt{1-4\lambda})n_{12}}{2})e_1+n_{12}e_2+  n_{13}e_3 ,~~ N_{D_{3}}^1(e_2)= \frac{(-1+\sqrt{1-4\lambda})(n_{22}-n_{33})}{2}e_1+ n_{22}e_2+ n_{23}e_3 ,~~ N_{D_{3}}^1(e_3)=  n_{33}e_3,~n_{12},~\lambda \neq0.$
      \item[]  $N_{D_{3}}^2(e_1)=(n_{33}+\frac{(-1-\sqrt{1-4\lambda})n_{12}}{2})e_1+n_{12}e_2+  n_{13}e_3
          ,~~ N_{D_{3}}^2(e_2)= \frac{(-1-\sqrt{1-4\lambda})(n_{22}-n_{33})}{2}e_1+ n_{22}e_2+ n_{23}e_3
           ,~~ N_{D_{3}}^2(e_3)=  n_{33}e_3,~n_{12},~\lambda \neq0.$
          \item[]  $N_{D_{3}}^3(e_1)=n_{33}e_1+  n_{13}e_3
          ,~~ N_{D_{3}}^3(e_2)= \frac{(-1\pm \sqrt{1-4\lambda})(n_{22}-n_{33})}{2}e_1+ n_{22}e_2+ n_{23}e_3
           ,~~ N_{D_{3}}^3(e_3)=  n_{33}e_3,~\lambda \neq 0.$
   \end{itemize}

$(4)$~The Nijenhuis operators on the associative algebra~$(D_{4},\cdot)$~are:
\begin{itemize}
    \item[]  $N_{D_{4}}^1(e_1)=n_{11}e_1,~~ N_{D_{4}}^1(e_2)= n_{22}e_2+ n_{23}e_3 ,~~ N_{D_{4}}^1(e_3)=  n_{32}e_2+ n_{33}e_3.$
         \item[]  $N_{D_{4}}^2(e_1)=n_{11}e_1 ,~~ N_{D_{4}}^2(e_2)= n_{22}e_2 ,~~ N_{D_{4}}^2(e_3)=n_{31}e_1+n_{32}e_2 + n_{11}e_3,~n_{31}\neq0.$
         \item[]  $N_{D_{4}}^3(e_1)=n_{11}e_1 ,~~ N_{D_{4}}^3(e_2)=  n_{21}e_1+ n_{22}e_2 ,
             ~~ N_{D_{4}}^3(e_3)= n_{31}e_1+ n_{32}e_2+  n_{11}e_3,~n_{21}\neq0.$
          \item[]  $N_{D_{4}}^4(e_1)=n_{11}e_1+  n_{12}e_2 ,~~ N_{D_{4}}^4(e_2)= n_{22}e_2 ,
             ~~ N_{D_{4}}^4(e_3)=  n_{32}e_2 + n_{22}e_3,~n_{12}\neq 0.$
   \end{itemize}

$(5)$~The Nijenhuis operators on the associative algebra~$(D_{5},\cdot)$~are:
\begin{itemize}
    \item[]  $N_{D_{5}}^1(e_1)=n_{11}e_1,~~ N_{D_{5}}^1(e_2)= n_{22}e_2+ n_{23}e_3 ,~~ N_{D_{5}}^1(e_3)=  n_{32}e_2+ n_{33}e_3.$
         \item[]  $N_{D_{5}}^2(e_1)=n_{11}e_1 ,~~ N_{D_{5}}^2(e_2)= n_{22}e_2 ,~~ N_{D_{5}}^2(e_3)=n_{31}e_1+n_{32}e_2 + n_{11}e_3,~n_{31}\neq0.$
         \item[]  $N_{D_{5}}^3(e_1)=n_{11}e_1 ,~~ N_{D_{5}}^3(e_2)=  n_{21}e_1+ n_{22}e_2 ,
             ~~ N_{D_{5}}^3(e_3)= n_{31}e_1+ n_{32}e_2+  n_{11}e_3,~n_{21}\neq0.$
          \item[]  $N_{D_{5}}^4(e_1)=n_{11}e_1+  n_{12}e_2 ,~~ N_{D_{5}}^4(e_2)= n_{22}e_2 ,
             ~~ N_{D_{5}}^4(e_3)=  n_{32}e_2 + n_{22}e_3,~n_{12}\neq 0.$
   \end{itemize}

$(6)$~The Nijenhuis operators on the associative algebra~$(D_{6},\cdot)$~are:
\begin{itemize}
     \item[]  $N_{D_{6}}^1(e_1)=n_{11}e_1,~~ N_{D_{6}}^1(e_2)= n_{22}e_2+ n_{23}e_3 ,~~ N_{D_{6}}^1(e_3)=  n_{32}e_2+ n_{33}e_3.$
         \item[]  $N_{D_{6}}^2(e_1)=n_{11}e_1+n_{12}e_2  ,~~ N_{D_{6}}^2(e_2)= n_{11}e_2 ,~~ N_{D_{6}}^2(e_3)=n_{32}e_2 + n_{11}e_3,~n_{12}\neq0.$
         \item[]  $N_{D_{6}}^3(e_1)=n_{11}e_1 ,~~ N_{D_{6}}^3(e_2)=  n_{21}e_1+ n_{22}e_2+  n_{21}e_3 ,
             ~~ N_{D_{6}}^3(e_3)= n_{31}e_1+ n_{32}e_2+  (n_{11}+n_{31})e_3,~n_{21}\neq0.$
          \item[]  $N_{D_{6}}^4(e_1)=n_{11}e_1+  n_{13}e_3 ,~~ N_{D_{6}}^4(e_2)= n_{22}e_2 ,
             ~~ N_{D_{6}}^4(e_3)= (n_{11}-n_{13})e_3,~n_{13}\neq 0.$
           \item[]  $N_{D_{6}}^5(e_1)=n_{11}e_1,~~ N_{D_{6}}^5(e_2)= n_{22}e_2 ,
             ~~ N_{D_{6}}^5(e_3)=n_{31}e_1+n_{32}e_2+ (n_{11}+n_{31})e_3,~n_{31}\neq 0.$
               \item[]  $N_{D_{6}}^6(e_1)=n_{11}e_1+n_{12}e_2+n_{13}e_2,~~ N_{D_{6}}^6(e_2)= n_{22}e_2 ,
             ~~ N_{D_{6}}^6(e_3)=(\frac{n_{11}n_{12}-n_{12}n_{13}-n_{12}
             n_{22}}{n_{13}})e_2+ (n_{11}-n_{13})e_3,~n_{12},~n_{13}\neq 0.$
   \end{itemize}

$(7)$~The Nijenhuis operators on the associative algebra~$(D_{7},\cdot)$~are:
\begin{itemize}
   \item[]  $N_{D_{7}}^1(e_1)=n_{11}e_1,~~ N_{D_{7}}^1(e_2)= n_{22}e_2+ n_{23}e_3 ,~~ N_{D_{7}}^1(e_3)=  n_{32}e_2+ n_{33}e_3.$
         \item[]  $N_{D_{7}}^2(e_1)=n_{11}e_1+n_{12}e_2  ,~~ N_{D_{7}}^2(e_2)= n_{11}e_2 ,~~ N_{D_{7}}^2(e_3)=n_{32}e_2 + n_{11}e_3,~n_{12}\neq0.$
         \item[]  $N_{D_{7}}^3(e_1)=n_{11}e_1 ,~~ N_{D_{7}}^3(e_2)=  n_{21}e_1+ n_{22}e_2+  n_{21}e_3 ,
             ~~ N_{D_{7}}^3(e_3)= n_{31}e_1+ n_{32}e_2+  (n_{11}+n_{31})e_3,~n_{21}\neq0.$
          \item[]  $N_{D_{7}}^4(e_1)=n_{11}e_1+  n_{13}e_3 ,~~ N_{D_{7}}^4(e_2)= n_{22}e_2 ,
             ~~ N_{D_{7}}^4(e_3)= (n_{11}-n_{13})e_3,~n_{13}\neq 0.$
           \item[]  $N_{D_{7}}^5(e_1)=n_{11}e_1,~~ N_{D_{7}}^5(e_2)= n_{22}e_2 ,
             ~~ N_{D_{7}}^5(e_3)=n_{31}e_1+n_{32}e_2+ (n_{11}+n_{31})e_3,~n_{31}\neq 0.$
               \item[]  $N_{D_{7}}^6(e_1)=n_{11}e_1+n_{12}e_2+n_{13}e_2,~~ N_{D_{7}}^6(e_2)= n_{22}e_2 ,
             ~~ N_{D_{7}}^6(e_3)=(\frac{n_{11}n_{12}-n_{12}n_{13}-n_{12}
             n_{22}}{n_{13}})e_2+ (n_{11}-n_{13})e_3,~n_{12},~n_{13}\neq 0.$
   \end{itemize}

$(8)$~The Nijenhuis operators on the associative algebra~$(D_{8},\cdot)$~are:
\begin{itemize}
  \item[]  $N_{D_{8}}^1(e_1)=n_{11}e_1,~~ N_{D_{8}}^1(e_2)= n_{22}e_2+ n_{23}e_3 ,~~ N_{D_{8}}^1(e_3)=  n_{32}e_2+ n_{33}e_3.$
         \item[]  $N_{D_{8}}^2(e_1)=n_{11}e_1 ,~~ N_{D_{8}}^2(e_2)= n_{22}e_2 ,~~ N_{D_{8}}^2(e_3)=n_{31}e_1+n_{32}e_2 + n_{11}e_3,~n_{31}\neq0.$
         \item[]  $N_{D_{8}}^3(e_1)=n_{11}e_1 ,~~ N_{D_{8}}^3(e_2)=  n_{21}e_1+ n_{22}e_2 ,
             ~~ N_{D_{8}}^3(e_3)= n_{31}e_1+ n_{32}e_2+  n_{11}e_3,~n_{21}\neq0.$
          \item[]  $N_{D_{8}}^4(e_1)=n_{11}e_1+  n_{12}e_2 ,~~ N_{D_{8}}^4(e_2)= n_{22}e_2 ,
             ~~ N_{D_{8}}^4(e_3)=  n_{32}e_2 + n_{22}e_3,~n_{12}\neq 0.$
   \end{itemize}

$(9)$~The Nijenhuis operators on the associative algebra~$(D_{9},\cdot)$~are:
\begin{itemize}
    \item[]  $N_{D_{9}}^1(e_1)=n_{11}e_1,~~ N_{D_{9}}^1(e_2)= n_{22}e_2+ n_{23}e_3 ,~~ N_{D_{9}}^1(e_3)=  n_{32}e_2+ n_{33}e_3.$
         \item[]  $N_{D_{9}}^2(e_1)=(n_{31}+n_{33})e_1 ,~~ N_{D_{9}}^2(e_2)= n_{22}e_2+ n_{23}e_3 ,~~ N_{D_{9}}^2(e_3)=n_{31}e_1+n_{32}e_2 + n_{33}e_3,~n_{31}\neq0.$
         \item[]  $N_{D_{9}}^3(e_1)=(n_{31}+n_{33})e_1,~~ N_{D_{9}}^3(e_2)=n_{21}e_1+  n_{22}e_2- n_{21}e_3,
             ~~ N_{D_{9}}^3(e_3)= n_{31}e_1+n_{32}e_2 + n_{33}e_3,~n_{21}\neq0.$
          \item[]  $N_{D_{9}}^4(e_1)=n_{11}e_1+  n_{12}e_2 ,~~ N_{D_{9}}^4(e_2)= n_{11}e_2 ,
             ~~ N_{D_{9}}^4(e_3)= n_{32}e_2+n_{11}e_3,~n_{12}\neq 0.$
               \item[]  $N_{D_{9}}^5(e_1)=n_{11}e_1+n_{12}e_2+n_{13}e_2,~~ N_{D_{9}}^5(e_2)= n_{22}e_2 ,
             ~~ N_{D_{9}}^5(e_3)=(\frac{n_{11}n_{12}+n_{12}n_{13}-n_{12}
             n_{22}}{n_{13}})e_2+ (n_{11}+n_{13})e_3,~n_{13}\neq 0.$
   \end{itemize}

$(10)$~The Nijenhuis operators on the associative algebra~$(D_{10},\cdot)$~are:
\begin{itemize}
  \item[]  $N_{D_{10}}^1(e_1)=n_{11}e_1 +n_{12}e_2+n_{13}e_3,
   ~~ N_{D_{10}}^1(e_2)= n_{21}e_1 +n_{22}e_2+n_{23}e_3,~~
   N_{D_{10}}^1(e_3)=  n_{31}e_1 +n_{32}e_2+n_{33}e_3.$
   \end{itemize}

$(11)$~The Nijenhuis operators on the associative algebra~$(D_{11},\cdot)$~are:
\begin{itemize}
   \item[]  $N_{D_{11}}^1(e_1)=n_{11}e_1 +n_{12}e_2+n_{13}e_3,
   ~~ N_{D_{11}}^1(e_2)= n_{21}e_1 +n_{22}e_2+n_{23}e_3,~~
   N_{D_{11}}^1(e_3)=  n_{31}e_1 +n_{32}e_2+n_{33}e_3.$
   \end{itemize}

$(12)$~The Nijenhuis operators on the associative algebra~$(D_{12},\cdot)$~are:
\begin{itemize}
   \item[]  $N_{D_{12}}^1(e_1)=n_{11}e_1,
   ~~ N_{D_{12}}^1(e_2)= n_{22}e_2,~~
   N_{D_{12}}^1(e_3)=  n_{31}e_1 +n_{32}e_2+n_{33}e_3.$
     \item[]  $N_{D_{12}}^2(e_1)=n_{11}e_1,
   ~~ N_{D_{12}}^2(e_2)= n_{21}e_1 +n_{22}e_2,~~
   N_{D_{12}}^2(e_3)=  n_{31}e_1 +n_{32}e_2+n_{11}e_3,~n_{21}\neq 0.$
     \item[]  $N_{D_{12}}^3(e_1)=n_{11}e_1+n_{12}e_2,
   ~~ N_{D_{12}}^3(e_2)= n_{22}e_2,~~
   N_{D_{12}}^3(e_3)=  n_{31}e_1 +n_{32}e_2+n_{22}e_3,~n_{12}\neq 0.$
    \item[]  $N_{D_{12}}^4(e_1)=n_{11}e_1,
   ~~ N_{D_{12}}^4(e_2)= n_{21}e_1 +n_{22}e_2+n_{23}e_3,~~
   N_{D_{12}}^4(e_3)=  \frac{n_{21}n_{33}-n_{11}n_{21}}{n_{23}}e_1 +n_{32}e_2+n_{33}e_3,~n_{23}\neq 0.$
    \item[]  $N_{D_{12}}^5(e_1)=n_{11}e_1 +n_{12}e_2+n_{13}e_3,
   ~~ N_{D_{12}}^5(e_2)=n_{22}e_2,~~
   N_{D_{12}}^5(e_3)=n_{31}e_1 + \frac{n_{12}n_{33}-n_{12}n_{22}}{n_{13}}~e_2+n_{33}e_3,~n_{13}\neq 0.$
   \end{itemize}
\end{theorem}
\begin{proof}
The proof follows a similar approach to Proposition~\ref{Ni ope on 3 c a }.
\end{proof}
\begin{remark}
Theorem \ref{Ni ope on 3 c a } and Theorem \ref{Ni ope on 3 non c a } encompass all possible Nijenhuis operators on 3-dimensional associative algebras, and the above results are summarized in Tables \ref{tab:nijenhuis_3d_associative} $\sim$ \ref{tab:nijenhuis_3d_associative_D7-D12} of the Appendix. 
\end{remark}

\begin{remark}
A full orbit analysis under $\operatorname{Aut}(A)$ for each algebra is beyond the scope of this paper. As an illustration, consider $A_5$ in Proposition \ref{commutative pre-Lie algebra}. Its Nijenhuis operators are  
\[
N(e_1) = n_{11}e_1 + n_{12}e_2,\quad N(e_2) = n_{11}e_2,\quad n_{11}, n_{12} \in \mathbb{C}.
\]  
The automorphism group $\operatorname{Aut}(A_5) = \{\varphi \mid \varphi(e_1) = a e_1 + b e_2,\, \varphi(e_2) = a^2 e_2,\, a \ne 0,\, b \in \mathbb{C}\}$ acts by conjugation:  
\[
\varphi \circ N \circ \varphi^{-1} \colon N \mapsto \bigl(n_{11} \mapsto n_{11},\; n_{12} \mapsto a^{-1}n_{12}\bigr).
\]  
The parameter $b$ does not affect the action, so only $a$ rescales $n_{12}$. This yields two non-isomorphic types:  
 $(n_{12} = 0)$: $N(e_1) = n_{11}e_1,\; N(e_2) = n_{11}e_2$;  
$(n_{12} \ne 0)$: by setting $a = n_{12}$, we obtain the canonical form $N(e_1) = n_{11}e_1 + e_2,\; N(e_2) = n_{11}e_2$.  

These types are not isomorphic because the $\operatorname{Aut}(A_5)$-action preserves the condition $n_{12} = 0$. Similar reductions can be applied to other families. We present the operators in parametric form, where the parameters in a family may represent fewer distinct isomorphism classes than initially apparent.
\end{remark}

\section{Applications of Nijenhuis Operators to the CYBE}
In previous sections, we obtained Nijenhuis operators $N$ on 2-dimensional pre-Lie algebras and 3-dimensional commutative and non-commutative associative algebras $(A,\cdot)$. Using these results, we construct Rota-Baxter operators $R$ of weight zero on the sub-adjacent Lie algebras ${\frak g}(A)$ in this section. Consequently, these operators induce solutions of the CYBE on the Lie algebras ${\frak g}(A) \ltimes_{{\rm ad}^{\ast}} {\frak g}(A)^{\ast}$.

\subsection{From Nijenhuis operators to Rota-Baxter operators and CYBE solutions}

Briefly recall the construction:  
\begin{enumerate}  
    \item For a pre-Lie algebra $(A,\cdot)$, its sub-adjacent Lie algebra $\mathfrak{g}(A) = (A, [\cdot,\cdot])$ is defined by $[x,y] = x\cdot y - y\cdot x$.  
    \item A Nijenhuis operator $N$ on $A$ induces a Rota-Baxter operator $R$ of weight zero on $\mathfrak{g}(A)$ via Proposition~\ref{Ni ope and rb ope} and Proposition~\ref{rb on pre-Lie and Lie}.  
    \item A Rota-Baxter operator $R$ on $\mathfrak{g}(A)$ yields a solution $r$ of the CYBE on the double $\mathfrak{g}(A) \ltimes_{\mathrm{ad}^*} \mathfrak{g}(A)^*$ via the standard construction (Theorem~\ref{Bai:classical}).  
\end{enumerate}  
This provides a systematic pipeline:  
\[
\text{Nijenhuis on } A \xrightarrow{\text{Prop.~\ref{Ni ope and rb ope},~\ref{rb on pre-Lie and Lie}}} \text{Rota-Baxter on } \mathfrak{g}(A) \xrightarrow{\text{Thm.~\ref{Bai:classical}}} \text{CYBE solution on } \mathfrak{g}(A) \ltimes_{\mathrm{ad}^*} \mathfrak{g}(A)^*
\]

\begin{proposition}\label{Ni ope and rb ope}
Let  $\left(A, \cdot\right)$  be a pre-Lie algebra and  $N: A \rightarrow A$ be  a linear map.
If  $N^{2}=0$, then  $N$  is a Nijenhuis operator if and only if~$N$~is a Rota-Baxter operator of weight zero on  $(A,\cdot)$.
\end{proposition}
\begin{proof}
Let $N: \mathfrak{g} \rightarrow \mathfrak{g}$ be a linear map such that $N^2 = 0$.~If~$N$~is a Nijenhuis operator,~then
\begin{eqnarray*}
N(x) \cdot N(y)&=&N(N(x) \cdot  y+x \cdot N(y)-N(x \cdot  y))\\
&=&N(N(x) \cdot  y+x \cdot N(y))-N^2(x \cdot  y)\\
&=&N(N(x) \cdot  y+x \cdot N(y)),\quad\forall~ x, y \in A.
\end{eqnarray*}
Therefore~$N$~is a Rota-Baxter operator of weight zero on  $(A,\cdot)$.~
Conversely, if~$N$~is a Rota-Baxter operator of weight zero on  $(A,\cdot)$ such that $N^2 = 0$,  then $N$  is a Nijenhuis operator.
\end{proof}

\begin{proposition}\label{rb on pre-Lie and Lie}
If $R$ is a Rota-Baxter operator of weight zero on a pre-Lie algebra $(A,\cdot)$, then $R$ is
a Rota-Baxter operator of weight zero on its sub-adjacent Lie algebra ${\frak g}(A)$.
\end{proposition}
\begin{proof}
Let $R$ be a Rota-Baxter operator of weight zero on a pre-Lie algebra $(A,\cdot)$. Then for all~$x, y \in A$,~we
have
\begin{eqnarray*}
 R(x)\cdot R(y)=R(R(x)\cdot y + x\cdot R(y)),\\
 R(y)\cdot R(x)=R(R(y)\cdot x + y\cdot R(x)).
\end{eqnarray*}
Therefore,
\begin{eqnarray*}
[R(x), R(y)]&=&R(x)\cdot R(y)-R(y)\cdot R(x)\\
&=&R(R(x)\cdot y + x\cdot R(y))-R(R(y)\cdot x + y\cdot R(x))\\
&=& R(R(x)\cdot y - y\cdot R(x) + x\cdot R(y)-R(y)\cdot x)\\
&=&R([R(x), y] + [x, R(y)]),\quad\forall~x, y \in A.
\end{eqnarray*}
Hence $R$~is a Rota-Baxter operator of weight zero on its sub-adjacent Lie algebra ${\frak g}(A).$
\end{proof}

Let $\rho: \frak g \rightarrow \gl(V)$ be a representation of a Lie algebra $\frak g$.
On the vector space $\frak g \oplus V$,
we can define a Lie algebra structure
(denoted by $\frak g \ltimes_{\rho} V$ ) given by
\begin{equation}\label{Lie algebra structure}
 [x_{1} + v_{1}, x_{2} + v_{2}] = [x_{1}, x_{2}] + \rho(x_{1})v_{2} -
\rho(x_{2})v_{1},\quad \forall ~x_{1}, x_{2} \in \frak g, v_{1}, v_{2} \in V.
\end{equation}
Define a linear map $\rho^{\ast}: {\frak g} \rightarrow \gl({\frak g}^{\ast})$~by
\begin{equation}\label{dual rep}
\langle \rho_x^*(\xi),y \rangle=-\langle\xi, \rho_x(y)\rangle, \;\;
\forall~ x, y\in \g, \xi\in {\frak g}^*.
\end{equation}
Then~$\rho^{\ast}$~is a representation of a Lie algebra $\frak g$.~It is
the dual representation of~$\rho$.
\begin{proposition}\label{element in G+V}\rm{(\cite{Bai})}
If $R : V \rightarrow \mathfrak{g}$ is a linear map, then $R$ can be viewed
as an element in $\mathfrak{g}\otimes V^{\ast}
\subset (\mathfrak{g} \ltimes_{\rho^{\ast}} V^{\ast} ) \otimes
(\mathfrak{g} \ltimes_{\rho^{\ast}} V^{\ast} )$.
\end{proposition}
\begin{proof}
Let $\{e_{1},\cdots , e_{n}\}$ be a basis of $\mathfrak{g}$.
Let $\{v_{1}, \cdots , v_{m}\}$ be a basis of $V$ and
$\{v^{\ast}_{1},\cdots, v^{\ast}_{m}\}$ be its dual basis,
that is $v^{\ast}_{i} (v_{j}) = \delta_{ij}$.
Set $R(v_{i}) =\sum\limits_{j=1}^{n} a_{ij}e_{j},$ $i=1,2,\cdots, n$.
Since as vector spaces, ${\rm Hom}(V, \mathfrak{g}) \cong \mathfrak{g}
\otimes V^{\ast}$, we have
\begin{equation}
R = \sum_{i=1}^{m} R(v_{i}) \otimes v_{i}^{\ast} =
\sum_{i=1}^{m} \sum_{j=1}^{n} a_{ij} e_{j} \otimes v_{i}^{\ast}
\in \mathfrak{g} \otimes V^{\ast} \subset
(\mathfrak{g} \ltimes_{\rho^{\ast}} V^{\ast} ) \otimes
(\mathfrak{g} \ltimes_{\rho^{\ast}} V^{\ast} ).
\end{equation}

Hence $R$ can be viewed
as an element in $\mathfrak{g}\otimes V^{\ast}
\subset (\mathfrak{g} \ltimes_{\rho^{\ast}} V^{\ast} ) \otimes
(\mathfrak{g} \ltimes_{\rho^{\ast}} V^{\ast} )$.
\end{proof}

For any tensor element $r=\sum_i a_i\otimes b_i\in V \otimes V$, denote $r^{21}=\sum_ib_i\otimes a_i$.

\begin{theorem}\rm{(\cite{Bai})}\label{Bai:classical}
Let $\frak g$ be a Lie algebra. A linear map $R: \frak g\rightarrow \frak g$ is a Rota-Baxter operator if and only if
$r = R - R^{21}$ is a skew-symmetric solution of CYBE
in $\mathfrak{g} \ltimes_{{\rm ad}^{\ast}} \frak g^{\ast}$.
\end{theorem}

\subsection{A complete worked example: from a 2D pre-Lie algebra $(B_1,\cdot)$ to an explicit $r$-matrix}
\label{subsec:example}
We now illustrate the pipeline from Nijenhuis operators to CYBE solutions through a concrete, end-to-end computation.

\textbf{Step 1: Choice of a pre-Lie algebra and its Nijenhuis operator.}
Consider the 2-dimensional pre-Lie algebra $(B_1,\cdot)$ with basis $\{e_1, e_2\}$ (appearing as algebra $A_{-1}$ in \cite{Bai-Meng}) and non-zero products  
\begin{equation}\label{eq:B1-product}  
e_2 \cdot e_1 = -e_1, \qquad e_2 \cdot e_2 = e_1 - e_2.  
\end{equation}  
From Theorem \ref{Ni ope on 2 non c}, the Nijenhuis operators on $(B_1,\cdot)$ are given by the family $N_{B_1}^1$ with two free complex parameters $n_{11}, n_{21} \in \mathbb{C}$:  
\begin{equation}\label{eq:Nij-B1}  
N_{B_1}^1(e_1) = n_{11} e_1, \qquad N_{B_1}^1(e_2) = n_{21} e_1 + n_{11} e_2,  
\end{equation}  
or in matrix form (relative to $\{e_1, e_2\}$)  
\[
N_{B_1}^1 = \begin{pmatrix} n_{11} & 0 \\ n_{21} & n_{11} \end{pmatrix}.
\]

\textbf{Step 2: Sub-adjacent Lie algebra.}
The sub-adjacent Lie algebra $\mathfrak{g}(B_1) = (B_1, [\cdot,\cdot])$ is obtained via $[x,y] = x\cdot y - y\cdot x$. Using~\eqref{eq:B1-product},  
\begin{align*}
[e_1, e_2] &= e_1\cdot e_2 - e_2\cdot e_1 = 0 - (-e_1) = e_1, \\
[e_2, e_1] &= e_2\cdot e_1 - e_1\cdot e_2 = -e_1 - 0 = -e_1,
\end{align*}  
with all other brackets zero. Hence, the only non-trivial brackets are  
$$
[e_1, e_2] = e_1, \quad [e_2, e_1] = -e_1.
$$

\textbf{Step 3: From Nijenhuis to Rota-Baxter operator.}
First compute the square of $N_{B_1}^1$:  
\[
(N_{B_1}^1)^2(e_1) = n_{11}^2 e_1, \qquad (N_{B_1}^1)^2(e_2) = 2n_{11}n_{21} e_1 + n_{11}^2 e_2.
\]  
According to Proposition~\ref{Ni ope and rb ope} and~\ref{rb on pre-Lie and Lie}, a Nijenhuis operator $N$ induces a weight-zero Rota-Baxter operator on $\mathfrak{g}(B_1)$ precisely when $N^2 = 0$. This forces  
\[
n_{11} = 0,
\]  
while $n_{21}$ remains free. Consequently, the admissible Nijenhuis operators are  
\[
N_{B_1}^1(e_1) = 0, \qquad N_{B_1}^1(e_2) = n_{21} e_1,
\]  
and the induced Rota-Baxter operator $R_{\mathfrak{g}(B_1)}^1$ on $\mathfrak{g}(B_1)$ has matrix  
\begin{equation}\label{eq:RB-param}
R_{\mathfrak{g}(B_1)}^1 = \begin{pmatrix} 0 & 0 \\ n_{21} & 0 \end{pmatrix}, \qquad \text{i.e.} \qquad R(e_1) = 0,\; R(e_2) = n_{21} e_1.
\end{equation}

\textbf{Step 4: Constructing the CYBE solution.}
Let $\{e_1^*, e_2^*\}$ be the dual basis in $\mathfrak{g}(B_1)^*$. The remaining brackets in $\mathfrak{g}(B_1) \ltimes_{\mathrm{ad}^*} \mathfrak{g}(B_1)^*$ are determined by the coadjoint representation $\mathrm{ad}^*: \mathfrak{g}(B_1) \to \mathfrak{gl}(\mathfrak{g}(B_1)^*)$, defined by  
\[
\langle \mathrm{ad}_x^*(\xi), y \rangle = -\langle \xi, [x, y] \rangle,\qquad x,y\in\mathfrak{g}(B_1),\; \xi\in\mathfrak{g}(B_1)^*.
\]  
A direct computation gives  
\[
\mathrm{ad}_{e_1}^*(e_1^*) = -e_2^*,\quad \mathrm{ad}_{e_1}^*(e_2^*) = 0,\quad \mathrm{ad}_{e_2}^*(e_1^*) = e_1^*,\quad \mathrm{ad}_{e_2}^*(e_2^*) = 0,
\]  
which yields the brackets  
\begin{equation}\label{eq:bracket-g1-double}
[e_1, e_2] = e_1,\qquad [e_1, e_1^*] = -e_2^*,\qquad [e_2, e_1^*] = e_1^*.
\end{equation}  
Following the standard construction (Theorem~\ref{Bai:classical}), a solution $r \in \mathfrak{g}(B_1) \otimes \mathfrak{g}(B_1)^*$ of the CYBE on the Lie algebra $\mathfrak{g}(B_1) \ltimes_{\mathrm{ad}^*} \mathfrak{g}(B_1)^*$ is obtained via  
\begin{equation}\label{eq:r-formula}
r = \sum_{i=1}^2 \bigl( R(e_i) \otimes e_i^* - e_i^* \otimes R(e_i) \bigr).
\end{equation}  
Substituting the explicit action~\eqref{eq:RB-param}:  
\begin{align*}
r &= \bigl( R(e_1) \otimes e_1^* - e_1^* \otimes R(e_1) \bigr) + \bigl( R(e_2) \otimes e_2^* - e_2^* \otimes R(e_2) \bigr) \\
&= \bigl( 0 \otimes e_1^* - e_1^* \otimes 0 \bigr) + \bigl( n_{21} e_1 \otimes e_2^* - e_2^* \otimes n_{21} e_1 \bigr) \\
&= n_{21} \bigl( e_1 \otimes e_2^* - e_2^* \otimes e_1 \bigr).
\end{align*}

\textbf{Step 5: Result.}
We thus obtain the one-parameter family of skew-symmetric $r$-matrices  
\begin{equation}\label{eq:r-explicit}  
r = n_{21} \; e_1 \wedge e_2^*,\quad \text{where}~e_1 \wedge e_2^* := e_1 \otimes e_2^* - e_2^* \otimes e_1,  
\end{equation}  
which is a solution of the classical Yang–Baxter equation on $\mathfrak{g}(B_1) \ltimes_{\mathrm{ad}^*} \mathfrak{g}(B_1)^*$.

\begin{remark}
The parameter $n_{21}$ is free; $n_{21}=0$ gives the trivial solution $r=0$, while any $n_{21}\neq 0$ yields a non‑trivial  solution. The whole chain is
    \[
    \text{pre‑Lie } B_1 \;\longrightarrow\; \text{Nijenhuis } N_{B_1}^1 \;\longrightarrow\; \text{Rota-Baxter } R_{\mathfrak{g}(B_1)}^1 \;\longrightarrow\; \text{CYBE solution } r
    \]
\end{remark}

\subsection{Sub-adjacent Lie algebras of  pre-Lie algebras and associative algebras}
Having illustrated the construction for a specific pre-Lie algebra $B_1$, we now list the sub-adjacent Lie algebras arising from all 2-dimensional pre-Lie and 3-dimensional associative algebras. These are computed directly following \textbf{Step 2}. The resulting Lie algebras will serve as the domains for the Rota-Baxter operators (and hence for the CYBE solutions) obtained via the Nijenhuis-operator pipeline.

\begin{proposition}\label{sub lie 2}
Let $ {\frak g}(B) $ be the sub-adjacent Lie algebra of a 2-dimensional pre-Lie algebra $ (B, \cdot) $. Then $ {\frak g}(B) $ has a basis $ \{e_1, e_2\} $ for which the non-zero product is one of the following:
\begin{eqnarray*}
  && {\frak g}(B_1):~[e_1,e_2]=e_1;\quad\quad{\frak g}(B_2):~[e_1,e_2]=e_1; \\
  && {\frak g}(B_3):~[e_1,e_2]=e_1;\quad\quad{\frak g}(B_4):~[e_1,e_2]=e_1; \\
  && {\frak g}(B_5):~[e_1,e_2]=e_1;\quad\quad{\frak g}(B_6):~[e_1,e_2]=e_2.
\end{eqnarray*}
\end{proposition}

By Proposition \ref{sub lie 2}, ${\frak g}(B_1)$ equals ${\frak g}(B_2)$, ${\frak g}(B_3)$, ${\frak g}(B_4)$, and ${\frak g}(B_5)$; denote this common sub-adjacent Lie algebra by ${\frak g_1}$. Keep the notation ${\frak g}(B_6)$.

\begin{proposition}\label{sub lie 3}
Let ${\frak g}(D)$ be a non-trivial sub-adjacent Lie algebra of $3$-dimensional associative algebras~$(D,\cdot)$. Then ${\frak g}(D)$ has a basis $\left\{e_{1}, e_{2},e_{3}\right\}$ in which the non-zero products are described by one of the following:
\begin{eqnarray*}
  &&{\frak g}(D_1):~[e_1,e_2]=e_3;\quad\quad{\frak g}(D_2):[e_1,e_2]=e_3;\quad\quad{\frak g}(D_3):~[e_1,e_2]=e_3;\\
  && {\frak g}(D_4):~ [e_2,e_3]=-e_2;\quad\quad{\frak g}(D_5):~[e_2,e_3]=e_2;\quad\quad{\frak g}(D_6): [e_2,e_3]=-e_2; \\
  && {\frak g}(D_7): ~[e_2,e_3]=e_2;\quad\quad{\frak g}(D_8): [e_2,e_3]=-e_2;\quad\quad{\frak g}(D_9):~[e_2,e_3]=-e_2;\\
  &&{\frak g}(D_{10}): [e_1,e_3]=-e_1,~[e_2,e_3]=-e_2;\quad\quad{\frak g}(D_{11}): ~[e_1,e_3]=e_1,~[e_2,e_3]=e_2;\\
  &&{\frak g}(D_{12}): [e_1,e_3]=-e_1,~[e_2,e_3]=e_2.
\end{eqnarray*}
\end{proposition}

By Proposition \ref{sub lie 3}, ${\frak g}(D_1) = {\frak g}(D_2) = {\frak g}(D_3)$, denoted ${\frak g_2}$. ${\frak g}(D_4) = {\frak g}(D_6) = {\frak g}(D_8) = {\frak g}(D_9)$, denoted ${\frak g_3}$. ${\frak g}(D_5) = {\frak g}(D_7)$, denoted ${\frak g_4}$. Keep the notations $\mathfrak{g}(D_{10})$, $\mathfrak{g}(D_{11})$ and $\mathfrak{g}(D_{12})$.

\begin{remark}
The sub-adjacent Lie algebras of the pre-Lie algebras~$(A,\cdot)$~and~$(C,\cdot)$~are trivial. Therefore, all linear maps on their sub-adjacent Lie algebras are Rota-Baxter operators of weight zero.
\end{remark}

\subsection{Rota-Baxter operators on sub-adjacent Lie algebras}
We use the $2\times 2$ matrix ~$R_{{\frak g}(A)}^j$~ to represent the ~$j$-th Rota-Baxter operator of weight zero on the sub-adjacent Lie algebra ~${\frak g}(A)$. The following results are obtained using a method similar to that in \textbf{Step 3}.
\begin{proposition}\label{thm:g1}
The Rota-Baxter operators of weight zero on ${\frak g_1}$, derived from the Nijenhuis operator $N_B$, are as follows:
\begin{center}
{\footnotesize \begin{tabular}{ll}\toprule[1.5pt]\specialrule{0em}{2pt}{2pt} {\normalsize Nijenhuis operators on $B$} & { \normalsize Rota-Baxter operators of weight zero on ${\frak g_1}$}
\\\midrule
$N_{B_1}^1,~N_{B_5}^1$ & $R_{\frak g_1}^1=\left(\begin{array}{cc}
  0 & 0 \\
  n_{21} & 0 \\
\end{array}\right)
$
\\\specialrule{0em}{1pt}{2pt}
$N_{B_3}^2$&$R_{\frak g_1}^2=0$
\\\specialrule{0em}{2pt}{2pt}
 $N_{B_2}^1,~N_{B_4}^1$ & $R_{\frak g_1}^3=\{ N|N^2=0\}$
\\\specialrule{0em}{1pt}{2pt}
$N_{B_3}^1$ & $R_{\frak g_1}^4=\left(\begin{array}{cc}
  0 & 0 \\
  n_{21} & 0 \\
\end{array}\right),~ n_{21}\neq0
$
\\\bottomrule[1.5pt]
\end{tabular}}
\end{center}
\end{proposition}

\begin{remark}
   The Rota-Baxter operators of weight zero on the sub-adjacent Lie algebra~${{\frak g}(B_{6})}$ derived by Nijenhuis operators $N_B$ is~0.
\end{remark}

\begin{proposition}\label{proof of g2}
The matrices of Rota-Baxter operators of weight zero on $\frak g_2$, induced by Nijenhuis operators $N$, are as follows:

\begin{center}
{\footnotesize \begin{tabular}{ll}\toprule[1.5pt]\specialrule{0em}{2pt}{2pt} {\normalsize Nijenhuis operators on $D$} & {\normalsize Rota-Baxter operators of weight zero on~${\frak g_2}$~}
\\\midrule
 $N_{D_1}^1$
& $R_{\frak g_2}^1=
\left(\begin{array}{ccc}0 & 0 &0\\
  n_{21} & 0&n_{23}\\
  0& 0& 0 \\
\end{array}\right)$\\\specialrule{0em}{1pt}{1pt}&
$R_{\frak g_2}^2=
\left(\begin{array}{ccc}0 & 0 &n_{13}\\
  0 & 0&n_{23}\\
  0& 0& 0 \\
\end{array}\right),~n_{13}\neq0$
\\\specialrule{0em}{2pt}{2pt}
$N_{D_1}^2,~N_{D_2}^2$
&$R_{\frak g_2}^3=
\left(\begin{array}{ccc}0 & n_{12} &n_{13}\\
  0 & 0&0\\
  0& 0& 0 \\
\end{array}\right),~n_{12}\neq0$
\\\specialrule{0em}{2pt}{2pt}
$N_{D_1}^3$&$R_{\frak g_2}^4=
\left(\begin{array}{ccc}n_{11} & n_{12} &-\frac{n_{12}n_{23}}{n_{11}}\\
 -\frac{{n_{11}}^2}{n_{12}} & -n_{11}&n_{23}\\
  0& 0& 0 \\
\end{array}\right),~ n_{11}, n_{12}\neq0$
\\\specialrule{0em}{2pt}{2pt}
$N_{D_2}^1$
&$R_{\frak g_2}^5=\left(\begin{array}{ccc}
0 & 0 &0\\
  n_{21} & 0&n_{23}\\
  0& 0& 0 \\
\end{array}\right),~ n_{21}\neq0 $
\\\specialrule{0em}{2pt}{2pt}
$N_{D_2}^3,~N_{D_3}^3$&$R_{\frak g_2}^6=
\left(\begin{array}{ccc}
0 & 0 &n_{13}\\
0 & 0&n_{23}\\
0& 0& 0 \\
\end{array}\right)$
\\\specialrule{0em}{2pt}{2pt}
$N_{D_3}^1$&$R_{\frak g_2}^7=
\left(\begin{array}{ccc}
k_1 n_{12} & n_{12} &n_{13}\\
-{k_1}^2n_{12} & -k_1n_{12}&-k_1n_{13}\\
0& 0& 0 \\
\end{array}\right),~\begin{array}{c}
k_1=\frac{(-1+\sqrt{1-4\lambda})}{2}, \\
n_{12},\lambda \neq0
\end{array}$
\\\specialrule{0em}{2pt}{2pt}
$N_{D_3}^2$&$R_{\frak g_2}^8=
\left(\begin{array}{ccc}
k_2 n_{12} & n_{12} &n_{13}\\
-{k_2}^2n_{12} & -k_2n_{12}&-k_2n_{13}\\
0& 0& 0 \\
\end{array}\right),~\begin{array}{c}
k_2=\frac{(-1-\sqrt{1-4\lambda})}{2}, \\
n_{12},\lambda \neq0
\end{array}$
\\\bottomrule[1.5pt]
\end{tabular}}
\end{center}
\end{proposition}

\begin{proof}
The Nijenhuis operators $N_{D_1}^1$ on the pre-Lie algebra $(D_1,\cdot)$ are
  $$N_{D_{1}}^1(e_1)=n_{11}e_1+  n_{13}e_3 ,~~ N_{D_{1}}^1(e_2)= n_{21}e_1+ n_{22}e_2+ n_{23}e_3 ,~~ N_{D_{1}}^1(e_3)=  n_{11}e_3.$$
 Then we have
  \begin{eqnarray*}
    ({N_{D_1}^1})^2(e_1) &=&{n_{11}}^2e_1+n_{11}n_{13}e_2,\\
    ({N_{D_1}^1})^2(e_2)&=& n_{21}(n_{11}+n_{22})e_1+ {n_{22}}^2e_2+(n_{13}n_{21}+n_{22}n_{23}+n_{11}n_{23})e_3, \\
    ({N_{D_1}^1})^2(e_3) &=& {n_{11}}^2e_3.
  \end{eqnarray*}
   If  $({N_{D_1}^1})^2=0$,~then~$n_{11}=0$,~$n_{22}=0$,~$n_{13}n_{21}=0$.

 \textbf{Case~1}: If~$n_{13}=0$,
by Proposition~\ref{Ni ope and rb ope}~and~\ref{rb on pre-Lie and Lie},  we obtain the Rota-Baxter operator~$R_{\frak g_2}^1$.

\textbf{Case~2}: If~$n_{13}\neq0$,
~then~$n_{21}=0$, 
 yielding the Rota-Baxter operator $R_{\frak g_2}^2$.
All other Rota-Baxter operators are obtained similarly. This completes the proof.
\end{proof}

The following Propositions \ref{thm:g3}, Propositions \ref{thm:g4}, and Propositions \ref{thm:g5} are direct applications of Theorem \ref{Ni ope on 3 non c a } and follow the same approach as in the proof of Proposition~\ref{proof of g2}.

\begin{proposition}\label{thm:g3}
The Rota-Baxter operators of weight zero on ${\frak g_3}$, induced by Nijenhuis operators $N$, are as follows:
\begin{center}
{\footnotesize \begin{tabular}{ll}\toprule[1.5pt]\specialrule{0em}{2pt}{2pt} {\normalsize
Nijenhuis operators on $D$} & {\normalsize Rota-Baxter operators of weight zero on~${\frak g_3}$~}
\\\midrule
$
N_{D_4}^1,~N_{D_6}^1,~N_{D_8}^1,~N_{D_9}^1$&$
R_{\frak g_3}^1=\left(\begin{array}{ccc}0 & 0&0\\
0 & 0&0\\
  0&  n_{32}&0 \\
\end{array}\right)$\\\specialrule{0em}{1pt}{1pt}
&$R_{\frak g_3}^2=
\left(\begin{array}{ccc}0 & 0&0\\
0 & n_{22}&n_{23}\\
  0&  -\frac{{n_{22}}^2}{n_{23}}&-n_{22} \\
\end{array}\right),~ n_{23}\neq0$
\\\specialrule{0em}{2pt}{2pt}
$
N_{D_4}^2,~N_{D_8}^2$&$
R_{\frak g_3}^3=\left(\begin{array}{ccc}
0 & 0 &0\\
  0 & 0&0\\
 n_{31}& n_{32}&0
\end{array}\right),~ n_{31}\neq0$
\\\specialrule{0em}{2pt}{2pt}
$
N_{D_4}^3,~N_{D_8}^3$&$
R_{\frak g_3}^4=\left(\begin{array}{ccc}
0 & 0 &0\\
  n_{21} & 0&0\\
 n_{31}& 0& 0
\end{array}\right),~ n_{21}\neq0$
\\\specialrule{0em}{2pt}{2pt}
$
N_{D_4}^4,~N_{D_6}^2,~N_{D_8}^4,~N_{D_9}^4$&$
R_{\frak g_3}^5=\left(\begin{array}{ccc}
0 & n_{12} &0\\
  0 & 0&0\\
 0& n_{32}&0
\end{array}\right),~n_{12}\neq0$
\\\specialrule{0em}{2pt}{2pt}
$
N_{D_6}^3$&$
R_{\frak g_3}^6=\left(\begin{array}{ccc}
 0 & 0 &0\\
  n_{21} & -n_{31}&n_{21}\\
 n_{31}& -\frac{{n_{31}}^2}{n_{21}}&  n_{31}
\end{array}\right),~n_{21}\neq0$
\\\specialrule{0em}{2pt}{2pt}
$N_{D_9}^3$&$R_{\frak g_3}^7=\left(\begin{array}{ccc}
0& 0 &0\\
  n_{21}& -n_{33}&-n_{21}\\
   -n_{33}&  \frac{{n_{33}}^2}{n_{21}}& n_{33}
\end{array}\right),~n_{21}\neq0$
\\\bottomrule[1.5pt]
\end{tabular}}
\end{center}
\end{proposition}

\begin{proposition}\label{thm:g4}
The Rota-Baxter operators of weight zero on ${\frak g_4}$, derived from Nijenhuis operators $N$, are as follows:
\begin{center}
{\footnotesize \begin{tabular}{ll}\toprule[1.5pt]\specialrule{0em}{2pt}{2pt}{\normalsize
Nijenhuis operators on $D$} & {\normalsize Rota-Baxter operators of weight zero on~${\frak g_4}$~}
\\\midrule
$
N_{D_5}^1,~N_{D_7}^1$&$
R_{\frak g_4}^1=\left(\begin{array}{ccc}0 & 0&0\\
0 & 0&0\\
  0&  n_{32}&0 \\
\end{array}\right)$\\\specialrule{0em}{1pt}{1pt}&
$R_{\frak g_4}^2=
\left(\begin{array}{ccc}0 & 0&0\\
0 & n_{22}&n_{23}\\
  0&  -\frac{{n_{22}}^2}{n_{23}}&-n_{22} \\
\end{array}\right),~ n_{23}\neq0$
\\\specialrule{0em}{2pt}{2pt}
$
N_{D_5}^2$&$
R_{\frak g_4}^3=\left(\begin{array}{ccc}
0 & 0 &0\\
  0 & 0&0\\
 n_{31}& n_{32}&0
\end{array}\right),~n_{31}\neq0$
\\\specialrule{0em}{2pt}{2pt}
$
N_{D_5}^3$&$
R_{\frak g_4}^4=\left(\begin{array}{ccc}
0 & 0 &0\\
  n_{21} & 0&0\\
 n_{31}& 0& 0
\end{array}\right),~n_{21}\neq0$
\\\specialrule{0em}{2pt}{2pt}
$
N_{D_5}^4,~N_{D_7}^2$&$
R_{\frak g_4}^5=\left(\begin{array}{ccc}
0 & n_{12} &0\\
  0 & 0&0\\
 0& n_{32}&0
\end{array}\right),~n_{12}\neq0$
\\\specialrule{0em}{2pt}{2pt}
$
N_{D_7}^3$&$
R_{\frak g_4}^4=\left(\begin{array}{ccc}
 0 & 0 &0\\
  n_{21} & -n_{31}&n_{21}\\
 n_{31}& -\frac{{n_{31}}^2}{n_{21}}&  n_{31}
\end{array}\right),~n_{21}\neq0$
\\\bottomrule[1.5pt]
\end{tabular}}
\end{center}
\end{proposition}

\begin{remark}
 All linear maps on~$(D_{10},\cdot)$~and~$(D_{11},\cdot)$, whose squares are 0, are Rota-Baxter operators of weight zero on the sub-adjacent Lie algebras~${{\frak g}(D_{10})}$~and~${{\frak g}(D_{11})}$.
\end{remark}

\begin{proposition}\label{thm:g5}
The Rota-Baxter operators of weight zero on ${\frak g}(D_{12})$ induced by Nijenhuis operators $N$ are as follows:
\begin{center}
{\footnotesize \begin{tabular}{ll}\toprule[1.5pt]\specialrule{0em}{2pt}{2pt} {\normalsize
Nijenhuis operators on $D$} & {\normalsize Rota-Baxter operators of weight zero on ${\frak g}(D_{12})$}
\\\midrule
$
N_{D_{12}}^1$&$
R_{{\frak g}_{(D_{12})}}^1=\left(\begin{array}{ccc}
 0 & 0 &0\\
  0 & 0&0\\
  n_{31} & n_{32}&0
\end{array}\right)$\\\specialrule{0em}{2pt}{2pt}
$
N_{D_{12}}^2$&$
R_{{\frak g}_{(D_{12})}}^2=\left(\begin{array}{ccc}
 0 & 0 &0\\
  n_{21} &0&0\\
  n_{31} & 0& 0
\end{array}\right),~n_{21}\neq0$
\\\specialrule{0em}{2pt}{2pt}
$
N_{D_{12}}^3$&$
R_{{\frak g}_{(D_{12})}}^3=\left(\begin{array}{ccc}
  0 & n_{12} &0\\
  0 & 0&0\\
  0 & n_{32}&0
\end{array}\right),~n_{12}\neq0$
\\\specialrule{0em}{2pt}{2pt}
$
N_{D_{12}}^4$&$
R_{{\frak g}_{(D_{12})}}^4=\left(\begin{array}{ccc}
 0 & 0 &0\\
  n_{21} & -n_{33}&n_{23}\\
  \frac{n_{21}n_{33}}{n_{23}} &-\frac{{n_{33}}^2}{n_{23}}& n_{33}
\end{array}\right),~n_{23}\neq0$
\\\specialrule{0em}{2pt}{2pt}
$
N_{D_{12}}^5$&$
R_{{\frak g}_{(D_{12})}}^5=\left(\begin{array}{ccc}
  -n_{33} & n_{12} &n_{13} \\
  0 & 0&0\\
  -\frac{{n_{33}}^2}{n_{13}}& \frac{n_{12}n_{33}}{n_{13}}& n_{33}
\end{array}\right),~n_{13}\neq0$
\\\bottomrule[1.5pt]
\end{tabular}}
\end{center}
\end{proposition}
\begin{remark}
There are the same Nijenhuis operators on different pre-Lie algebras.~For example,~$(1)~N_{D_{1}}^2$ is equal to $N_{D_{2}}^2$.~$(2)~N_{D_{4}}^1$ is equal to $N_{D_{5}}^1$,~$N_{D_{6}}^1$,~$N_{D_{7}}^1$,~$N_{D_{8}}^1$,~$N_{D_{9}}^1$.
~$(3)~N_{D_{4}}^2$ is equal to $N_{D_{5}}^2$,~$N_{D_{8}}^2$.
~$(4)~N_{D_{4}}^3$ is equal to $N_{D_{5}}^3$,~$N_{D_{8}}^3$.~$(5)~N_{D_{4}}^4$ is equal to $N_{D_{5}}^4$,~$N_{D_{8}}^4$.
~$(6)~N_{D_{6}}^2$ is equal to $N_{D_{7}}^2$,~$N_{D_{9}}^4$.
~$(7)~N_{D_{6}}^3$ is equal to $N_{D_{7}}^3$.
The sub-adjacent Lie algebras~${\frak g}(D_{4})$~and~${\frak g}(D_{12})$ are different,~but the Rota-Baxter operators derived by~$N_{D_{4}}^4$~and~$N_{D_{12}}^3$~are the same.
Therefore,~the Rota-Baxter operators of weight zero on the different sub-adjacent Lie algebras derived by Nijenhuis operators may be identical.

\end{remark}

\subsection{Solutions of CYBE on Lie algebras~${\frak g}(A) \ltimes_{{\rm ad}^{\ast}} {\frak g}(A)^{\ast}$}

\begin{proposition}\label{3-dim sol of cybe}
Let $\{e_{1}, e_{2}, e_{3}\}$ be a basis of ${\frak g_2}$ and $\{e^{\ast}_{1}, e^{\ast}_{2}, e^{\ast}_{3}\}$ be its dual basis. Then the Lie algebra ${\frak g_2} \ltimes_{{\rm ad}^{\ast}} {\frak g_2}^{\ast}$ has a basis $\{e_1,e_2,e_3,e_1^{\ast},e_2^{\ast},e^{\ast}_{3}\}$, with the following non-zero Lie bracket:  
$$
[e_1,e_2]=e_3,\quad [e_1,e^{\ast}_{3}]=-e^{\ast}_{2},\quad [e_2,e^{\ast}_{3}]=e^{\ast}_{1}.
$$  
The classical Yang–Baxter equation on this Lie algebra has solutions as follows:
\begin{itemize}
\item[] \hskip-0.5cm  $r_1= (n_{21}e_1+n_{23}e_2)\otimes e_2^{\ast}-e_2^{\ast}\otimes(n_{21}e_1+n_{23}e_2) .$
\item[]\hskip-0.5cm  $ r_2 = (n_{11}e_1+n_{12}e_2-\frac{n_{12}n_{23}}{n_{11}}e_3) \otimes e_1^{\ast}+(-\frac{{n_{11}}^2}{n_{12}}e_1-n_{11}e_2+n_{23}e_3) \otimes e_2^{\ast}\\
   - e_1^{\ast} \otimes (n_{11}e_1+n_{12}e_2-\frac{n_{12}n_{23}}{n_{11}}e_3)- e_2^{\ast}\otimes (-\frac{{n_{11}}^2}{n_{12}}e_1-n_{11}e_2+n_{23}e_3)~n_{11},n_{12}\neq0.$
\item[]\hskip-0.5cm  $ r_3=(n_{12}e_2+n_{13}e_3)\otimes e_1^{\ast}-e_1^{\ast}\otimes(n_{12}e_2+n_{13}e_3),~n_{12}\neq0.$
\item[]\hskip-0.5cm  $r_4=n_{13}e_3 \otimes e_1^{\ast}+n_{23}e_3 \otimes e_2^{\ast} - e_1^{\ast} \otimes n_{13}e_3- e_2^{\ast}\otimes n_{23}e_3.$
\item[]\hskip-0.5cm  $ r_5 = (\frac{(-1+\sqrt{1-4\lambda})}{2}n_{12}e_1+n_{12}e_2+n_{13}e_3) \otimes e_1^{\ast}-\frac{(-1+\sqrt{1-4\lambda})}{2}(\frac{(-1+\sqrt{1-4\lambda})
    }{2}n_{12}e_1+n_{12}e_2+n_{13}e_3) \otimes e_2^{\ast}\\
   - e_1^{\ast} \otimes (\frac{(-1+\sqrt{1-4\lambda})}{2}n_{12}e_1+n_{12}e_2+n_{13}e_3)+ e_2^{\ast}\otimes \frac{(-1+\sqrt{1-4\lambda})}{2}(\frac{(-1+\sqrt{1-4\lambda})}{2}n_{12}e_1
   +n_{12}e_2+n_{13}e_3),~n_{12},\lambda \neq0.$
\item[]\hskip-0.5cm  $ r_6 = (\frac{(-1-\sqrt{1-4\lambda})}{2}n_{12}e_1+n_{12}e_2+n_{13}e_3) \otimes e_1^{\ast}-\frac{(-1-\sqrt{1-4\lambda})}{2}(\frac{(-1-\sqrt{1-4\lambda})
    }{2}n_{12}e_1+n_{12}e_2+n_{13}e_3) \otimes e_2^{\ast}\\
   - e_1^{\ast} \otimes (\frac{(-1-\sqrt{1-4\lambda})}{2}n_{12}e_1+n_{12}e_2+n_{13}e_3)+ e_2^{\ast}\otimes \frac{(-1-\sqrt{1-4\lambda})}{2}(\frac{(-1-\sqrt{1-4\lambda})}{2}n_{12}e_1
   +n_{12}e_2+n_{13}e_3),~n_{12},\lambda\neq0.$
\end{itemize}
\end{proposition}
\begin{proof}
Let $\rho=\ad^{\ast}$.~By \eqref{dual rep},
the dual representation~$\ad^{\ast}: {\frak g_2} \rightarrow \gl({\frak g_2}^{\ast})$~is defined by
\begin{equation}
\langle \ad_x^*(\xi),y \rangle=-\langle\xi, \ad_x(y)\rangle, \;\;
\forall~x, y\in \g_2,~\xi \in {\frak g_2}^*.
\end{equation}
By (\ref{Lie algebra structure}) we have
\begin{equation*}
[e_1,e_1]=0,\quad[e_1,e_2]=-[e_2,e_1]=e_3,\quad[e_1,e_3]=-[e_3,e_1]=0,
\end{equation*}
\begin{equation*}
[e_2,e_2]=0,\quad[e_2,e_3]=-[e_3,e_2]=0,\quad[e_3,e_3]=0,
\end{equation*}
\begin{equation*}
[e^{\ast}_{1},e^{\ast}_{1}]=0,
\quad
[e^{\ast}_{1},e^{\ast}_{2}]=-[e^{\ast}_{2},e^{\ast}_{1}]=0,
\quad[e^{\ast}_{1},e^{\ast}_{3}]=-[e^{\ast}_{3},e^{\ast}_{1}]=0,
\end{equation*}
\begin{equation*}
[e^{\ast}_{2},e^{\ast}_{2}]=0,\quad
[e^{\ast}_{2},e^{\ast}_{3}]=-[e^{\ast}_{3},e^{\ast}_{2}]=0,
\quad[e^{\ast}_{3},e^{\ast}_{3}]=0,
\end{equation*}
\begin{equation*}
[e_1,e^{\ast}_{1}]=-[e^{\ast}_{1},e_1]=\ad_{e_1}^*e^{\ast}_{1},\quad
[e_1,e^{\ast}_{2}]=-[e^{\ast}_{2},e_1]=\ad_{e_1}^*e^{\ast}_{2},\quad
[e_1,e^{\ast}_{3}]=-[e^{\ast}_{3},e_1]=\ad_{e_1}^*e^{\ast}_{3},
\end{equation*}
\begin{equation*}
[e_2,e^{\ast}_{1}]=-[e^{\ast}_{1},e_2]=\ad_{e_2}^*e^{\ast}_{1},\quad
[e_2,e^{\ast}_{2}]=-[e^{\ast}_{2},e_2]=\ad_{e_2}^*e^{\ast}_{2},\quad
[e_2,e^{\ast}_{3}]=-[e^{\ast}_{3},e_2]=\ad_{e_2}^*e^{\ast}_{3},
\end{equation*}
\begin{equation*}
[e_3,e^{\ast}_{1}]=-[e^{\ast}_{1},e_3]=\ad_{e_3}^*e^{\ast}_{1},\quad
[e_3,e^{\ast}_{2}]=-[e^{\ast}_{2},e_3]=\ad_{e_3}^*e^{\ast}_{2}.\quad
[e_3,e^{\ast}_{3}]=-[e^{\ast}_{3},e_3]=\ad_{e_3}^*e^{\ast}_{3}.
\end{equation*}
Set~$\ad_{e_1}^*(e^{\ast}_{1})
=k_{11}e^{\ast}_{1}+k_{12}e^{\ast}_{2}+k_{13}e^{\ast}_{3},~k_{11},~k_{12},~k_{13}\in \mathbb{C}.$
Then
\begin{equation*}
\langle \ad_{e_1}^*(e^{\ast}_{1}),{e_1} \rangle=-\langle e^{\ast}_{1}, \ad_{e_1}(e_{1})\rangle=-\langle e^{\ast}_{1}, 0\rangle=0, \;\;
\end{equation*}
\begin{equation*}
(k_{11}e^{\ast}_{1}+k_{12}e^{\ast}_{2}+k_{13}e^{\ast}_{3})({e_1})=k_{11}.
\end{equation*}
This implies that $k_{11}=0$. From
\begin{equation*}
\langle \ad_{e_1}^*(e^{\ast}_{1}),{e_2} \rangle=-\langle e^{\ast}_{1}, \ad_{e_1}(e_{2})\rangle=-\langle e^{\ast}_{1}, {e_3}\rangle=0, \;\;
\end{equation*}
\begin{equation*}
(k_{11}e^{\ast}_{1}+k_{12}e^{\ast}_{2}+k_{13}e^{\ast}_{3})({e_2})=k_{12},
\end{equation*}
we get~$k_{12}=0$. From
\begin{equation*}
\langle \ad_{e_1}^*(e^{\ast}_{1}),{e_3} \rangle=-\langle e^{\ast}_{1}, \ad_{e_1}(e_{3})\rangle=-\langle e^{\ast}_{1}, 0\rangle=0, \;\;
\end{equation*}
\begin{equation*}
(k_{11}e^{\ast}_{1}+k_{12}e^{\ast}_{2}+k_{13}e^{\ast}_{3})({e_3})=k_{13},
\end{equation*}
we get~$k_{13}=0$.~Thus,~$\ad_{e_1}^*(e^{\ast}_{1})
=0$.
Similarly, we have $\ad_{e_1}^*(e^{\ast}_{2})=0$, $\ad_{e_1}^*(e^{\ast}_{3})=-e^{\ast}_{2}$, $\ad_{e_2}^*(e^{\ast}_{1})=0$, $\ad_{e_2}^*(e^{\ast}_{2})=0$, $\ad_{e_2}^*(e^{\ast}_{3})=e^{\ast}_{1}$, $\ad_{e_3}^*(e^{\ast}_{1})=0$, $\ad_{e_3}^*(e^{\ast}_{2})=0$, and $\ad_{e_3}^*(e^{\ast}_{3})=0$.
Therefore, the non-zero Lie brackets of ${\frak g_2} \ltimes_{{\rm ad}^{\ast}} {\frak g_2}^{\ast}$ are
$$
[e_1,e_2]=e_3,\quad [e_1,e^{\ast}_{3}]=-e^{\ast}_{2},\quad [e_2,e^{\ast}_{3}]=e^{\ast}_{1}.
$$

The Rota-Baxter operators~$R_{\frak g_2}^1$~of weight zero on the sub-adjacent Lie algebra~${\frak g_2}$~are
\begin{equation*}
R_{\frak g_2}^1=
\left(\begin{array}{ccc}0 & 0 &0\\
  n_{21} & 0&n_{23}\\
  0& 0& 0 \\
\end{array}\right).
\end{equation*}
By Proposition \ref{element in G+V}, we have $R_{\frak g_2}^1=(n_{21}e_1+n_{23}e_3) \otimes e_2^{\ast}$~and~$(R_{\frak g_1}^1)^{21}=e_2^{\ast} \otimes (n_{21}e_1+n_{23}e_3)$.
By Lemma \ref{Bai:classical}, the induced solutions of CYBE on the Lie algebra~${\frak g_2} \ltimes_{{\rm ad}^{\ast}} {\frak g_2}^{\ast}$~are
$$r_1 =R_{\frak g_2}^1-(R_{\frak g_2}^1)^{21}= (n_{21}e_1+n_{23}e_3) \otimes e_2^{\ast} - e_2^{\ast} \otimes (n_{21}e_1+n_{23}e_3).$$
The other solutions of CYBE on the Lie algebra~${\frak g_2} \ltimes_{{\rm ad}^{\ast}} {\frak g_2}^{\ast}$ are obtained for the same step. This completes the proof.
\end{proof}

Using a similar proof as in Proposition \ref{3-dim sol of cybe}, the solutions of CYBE on the 6-dimensional Lie algebras ${\frak g_3} \ltimes_{{\rm ad}^{\ast}} {\frak g_3}^{\ast}$, ${\frak g_4} \ltimes_{{\rm ad}^{\ast}} {\frak g_4}^{\ast}$, and ${\frak g}(D_{12}) \ltimes_{{\rm ad}^{\ast}} {\frak g}(D_{12})^{\ast}$ can be obtained by applying Rota-Baxter operators of weight zero on the sub-adjacent Lie algebras from Propositions \ref{thm:g3}, \ref{thm:g4}, and \ref{thm:g5}, as shown in the following Propositions \ref{pro:CYBEg3}, \ref{pro:CYBEg4}, and \ref{pro:CYBED12}.
\begin{proposition}\label{pro:CYBEg3}
Let $\{e_{1}, e_{2}, e_{3}\}$ be a basis of ${\frak g_3}$ and $\{e^{\ast}_{1}, e^{\ast}_{2}, e^{\ast}_{3}\}$ be its dual basis. Then the Lie algebra ${\frak g_3} \ltimes_{{\rm ad}^{\ast}} {\frak g_3}^{\ast}$ has a basis $\{e_1,e_2,e_3,e_1^{\ast},e_2^{\ast},e^{\ast}_{3}\}$, with the following non-zero Lie brackets:  
$$  
[e_2,e_3]=-e_2,\quad [e_2,e^{\ast}_{2}]=e^{\ast}_{3},\quad [e_3,e^{\ast}_{2}]=-e^{\ast}_{2}.  
$$  
The classical Yang-Baxter equation on this Lie algebra has solutions as follows:

\begin{itemize}
\item[] \hskip-0.5cm  $r_1 =n_{32}e_2 \otimes e_3^{\ast}-e_3^{\ast} \otimes n_{32}e_2 .$
\item[] \hskip-0.5cm  $r_2 = (n_{22}e_2+{n_{23}}e_3) \otimes e_2^{\ast}-(\frac{{n_{22}}^2}{n_{23}}e_2+n_{22}e_3) \otimes e_3^{\ast}\\
  -e_2^{\ast} \otimes (n_{22}e_2+{n_{23}}e_3)+ e_3^{\ast}\otimes (\frac{{n_{22}}^2}{n_{23}}e_2+n_{22}e_3),~n_{23}\neq0.$
\item[]\hskip-0.5cm  $ r_3=(n_{31}e_1+n_{32}e_2)\otimes e_3^{\ast}-e_3^{\ast}\otimes(n_{31}e_1+n_{32}e_2),~n_{31}\neq0.$
\item[]\hskip-0.5cm  $r_4=n_{21}e_1 \otimes e_2^{\ast}+n_{31}e_1 \otimes e_3^{\ast}-e_2^{\ast} \otimes n_{21}e_1- e_3^{\ast}\otimes n_{31}e_1,~n_{21}\neq0.$
\item[]\hskip-0.5cm  $r_5 =n_{12}e_2 \otimes e_1^{\ast}+n_{32}e_2 \otimes e_3^{\ast}-e_1^{\ast} \otimes n_{12}e_2- e_3^{\ast}\otimes n_{32}e_2,~n_{12}\neq0.$
\item[]\hskip-0.5cm  $r_6 =(n_{21}e_1-n_{31}e_2+{n_{21}}e_3) \otimes e_2^{\ast}+(n_{31}e_1-\frac{{n_{31}}^2}{n_{21}}e_2+n_{31}e_3) \otimes e_3^{\ast}\\
  -e_2^{\ast} \otimes (n_{21}e_1-n_{31}e_2+{n_{21}}e_3)- e_3^{\ast}\otimes (n_{31}e_1-\frac{{n_{31}}^2}{n_{21}}e_2+n_{31}e_3),~n_{21}\neq0.$
\item[]\hskip-0.5cm  $ r_7 = (n_{21}e_1-n_{33}e_2-{n_{21}}e_3) \otimes e_2^{\ast}+(-n_{33}e_1+\frac{{n_{33}}^2}{n_{21}}e_2+n_{33}e_3) \otimes e_3^{\ast}\\
  -e_2^{\ast} \otimes (n_{21}e_1-n_{33}e_2-{n_{21}}e_3)- e_3^{\ast}\otimes (-n_{33}e_1+\frac{{n_{33}}^2}{n_{21}}e_2+n_{33}e_3),~n_{21}\neq0.$
\end{itemize}
\end{proposition}

\begin{proposition}\label{pro:CYBEg4}
Let $\{e_{1}, e_{2}, e_{3}\}$ be a basis of ${\frak g_4}$ and  $\{e^{\ast}_{1}, e^{\ast}_{2}, e^{\ast}_{3}\}$ be its dual basis. Then the Lie algebra ${\frak g_4} \ltimes_{{\rm ad}^{\ast}} {\frak g_4}^{\ast}$ has a basis $\{e_1, e_2, e_3, e_1^{\ast}, e_2^{\ast}, e_3^{\ast}\}$, with the non-zero Lie brackets:  
$$  
[e_2, e_3] = e_2, \quad [e_2, e^{\ast}_{2}] = -e^{\ast}_{3}, \quad [e_3, e^{\ast}_{2}] = e^{\ast}_{2}.  
$$  
The classical Yang–Baxter equation on this Lie algebra has solutions as follows:

\begin{itemize}
\item[] \hskip-0.5cm  $r_1 =n_{32}e_2 \otimes e_3^{\ast}-e_3^{\ast} \otimes n_{32}e_2 .$
\item[] \hskip-0.5cm  $r_2=(n_{22}e_2+{n_{23}}e_3) \otimes e_2^{\ast}-(\frac{{n_{22}}^2}{n_{23}}e_2+n_{22}e_3) \otimes e_3^{\ast}\\
    -
e_2^{\ast} \otimes (n_{22}e_2+{n_{23}}e_3)+ e_3^{\ast}\otimes (\frac{{n_{22}}^2}{n_{23}}e_2+n_{22}e_3),~n_{23}\neq0.$
\item[]\hskip-0.5cm  $ r_3=(n_{31}e_1+n_{32}e_2)\otimes e_3^{\ast}-e_3^{\ast}\otimes(n_{31}e_1+n_{32}e_2),~n_{31}\neq0.$
\item[]\hskip-0.5cm  $r_4=n_{21}e_1 \otimes e_2^{\ast}+n_{31}e_1 \otimes e_3^{\ast}-e_2^{\ast} \otimes n_{21}e_1- e_3^{\ast}\otimes n_{31}e_1,~n_{21}\neq0.$
\item[]\hskip-0.5cm  $r_5=n_{12}e_2 \otimes e_1^{\ast}+n_{32}e_2 \otimes e_3^{\ast}-e_1^{\ast} \otimes n_{12}e_2- e_3^{\ast}\otimes n_{32}e_2,~n_{12}\neq0.$
\item[]\hskip-0.5cm  $r_6=(n_{21}e_1-n_{31}e_2+{n_{21}}e_3) \otimes e_2^{\ast}+(n_{31}e_1-\frac{{n_{31}}^2}{n_{21}}e_2+n_{31}e_3) \otimes e_3^{\ast}\\-
e_2^{\ast} \otimes (n_{21}e_1-n_{31}e_2+{n_{21}}e_3)- e_3^{\ast}\otimes (n_{31}e_1-\frac{{n_{31}}^2}{n_{21}}e_2+n_{31}e_3),~n_{21}\neq0.$
\end{itemize}
\end{proposition}

\begin{proposition} \label{pro:CYBED12}
Let $\{e_{1}, e_{2}, e_{3}\}$ be a basis of ${\frak g}(D_{12})$, and let $\{e^{\ast}_{1}, e^{\ast}_{2}, e^{\ast}_{3}\}$ be its dual basis. Then the Lie algebra ${\frak g}(D_{12}) \ltimes_{{\rm ad}^{\ast}} {\frak g}(D_{12})^{\ast}$ has a basis $\{e_1, e_2, e_3, e_1^{\ast}, e_2^{\ast}, e_3^{\ast}\}$, with the non-zero Lie brackets:  
$$
[e_1,e_3]=-e_1,\quad [e_1,e^{\ast}_{1}]=e^{\ast}_{3},\quad [e_2,e_3]=e_2,\quad [e_2,e^{\ast}_{2}]=-e^{\ast}_{3},\quad [e_3,e^{\ast}_{1}]=-e^{\ast}_{1},\quad [e_3,e^{\ast}_{2}]=e^{\ast}_{2}.
$$  
The classical Yang–Baxter equation on this Lie algebra has solutions as follows:
\begin{itemize}
\item[] \hskip-0.5cm  $r_1=(n_{31}e_1+n_{32}e_2)\otimes e_3^{\ast}-e_3^{\ast}\otimes(n_{31}e_1+n_{32}e_2).$
\item[]\hskip-0.5cm  $r_2=n_{21}e_1 \otimes e_2^{\ast}+n_{31}e_1 \otimes e_3^{\ast}-e_2^{\ast} \otimes n_{21}e_1- e_3^{\ast}\otimes n_{31}e_1,~n_{21}\neq0.$
\item[]\hskip-0.5cm  $r_3=n_{12}e_2 \otimes e_1^{\ast}+n_{32}e_2 \otimes e_3^{\ast}-e_1^{\ast} \otimes n_{12}e_2- e_3^{\ast}\otimes n_{32}e_2,~n_{12}\neq0.$
\item[]\hskip-0.5cm  $r_4=(n_{21}e_1-n_{33}e_2+{n_{23}}e_3) \otimes e_2^{\ast}+(\frac{n_{21}n_{33}}{n_{23}}e_1-\frac{{n_{33}}^2}{n_{23}}e_2+n_{33}e_3) \otimes e_3^{\ast}\\
    -e_2^{\ast} \otimes (n_{21}e_1-n_{33}e_2+{n_{23}}e_3)- e_3^{\ast}\otimes (\frac{n_{21}n_{33}}{n_{23}}e_1-\frac{{n_{33}}^2}{n_{23}}e_2+n_{33}e_3) ,~n_{23}\neq0.$
\item[]\hskip-0.5cm  $ r_5 = (-n_{33}e_1+n_{12}e_2+{n_{13}}e_3) \otimes e_1^{\ast}+(-\frac{{n_{33}}^2}{n_{13}}e_1+\frac{n_{12}n_{33}}{n_{13}}e_2+n_{33}e_3) \otimes e_3^{\ast}\\
  -e_1^{\ast} \otimes (-n_{33}e_1+n_{12}e_2+{n_{13}}e_3)- e_3^{\ast}\otimes (-\frac{{n_{33}}^2}{n_{13}}e_1+\frac{n_{12}n_{33}}{n_{13}}e_2+n_{33}e_3),~n_{13}\neq0.$
\end{itemize}
\end{proposition}

\begin{remark}[Relation to known CYBE solutions]\label{rem:novelty}
Unlike the direct computation of Rota-Baxter operators on simple three-dimensional $\mathfrak{sl}(2)$ (\cite{Pei}) or $\mathfrak{so}(3)$ (\cite{Jiang}), the solutions obtained in this subsection are constructed via Nijenhuis operators on three-dimensional non-commutative associative algebras $D_1$ to $D_{12}$. The sub-adjacent Lie algebras $\mathfrak{g}(D_1)$ to $\mathfrak{g}(D_{12})$ are solvable three-dimensional Lie algebras (they are not isomorphic to the solvable 3-dimensional Lie algebra in \cite{Linli-Wu}), from which we derive new Rota-Baxter operators and hence solutions to the classical Yang-Baxter equation on the corresponding semidirect products.
\end{remark}

\section{Conclusions}
This work fully describes Nijenhuis operators on 2-dimensional pre-Lie algebras and 3-dimensional associative algebras, and systematically constructs the corresponding Rota-Baxter operators (from  $N^2=0$) and solutions to the CYBE. We conclude with remarks on implications and future research directions.

(1) \textbf{Classification suggests for higher dimensions.}
The systematic three-step approach demonstrated in dimensions 2 and 3 can be generalized to any finite dimension:  
\begin{itemize}  
  \item Step 1: Classify the underlying high-dimensional algebras up to isomorphism (e.g., 4-dimensional nilpotent pre-Lie algebras \cite{Adashev}).  
  \item Step 2: Derive equations for the coefficients satisfying the Nijenhuis operator identity using its definition. For dimensions $\geq 4$, symbolic computation combined with theoretical analysis is essential, as manual calculation becomes impractical (or solve the system of equations \eqref{structural constants} in Proposition \ref{PRO: Structural constants}).  
  \item Step 3: Obtain the Rota-Baxter operator on the sub-adjacent Lie algebra from $N^2 = 0$. Following \cite[Section 2]{Bai}, the corresponding skew-symmetric solution of the Classical Yang-Baxter Equation (CYBE) on the Lie algebra $\mathfrak{g}(A) \ltimes_{\mathrm{ad}^*} \mathfrak{g}(A)^*$ can be constructed.  
\end{itemize}

(2) \textbf{Computational verification of the Nijenhuis operators.}
All operator families presented in Theorems~\ref{Ni ope on 2 c}, \ref{Ni ope on 2 non c}, 
\ref{Ni ope on 3 c a }, \ref{Ni ope on 3 non c a }  were verified symbolically using \textsf{Mathematica}~14. For each Nijenhuis family, we confirmed that $N$ satisfies the Nijenhuis equation:  
\begin{equation*}
N(x) \cdot N(y) = N(N(x) \cdot y + x \cdot N(y) - N(x \cdot y)), \quad \forall~x, y \in A.
\end{equation*}  
The code checks all families for arbitrary parameter values and is available from the authors upon request. This computational verification prevents algebraic errors from manual calculations and ensures the correctness of the results.

(3) \textbf{Future directions}. Our Nijenhuis family lists are parametrically complete but do not classify up to algebra automorphisms. Several natural extensions remain open:  
\begin{itemize}  
  \item A full orbit analysis under the automorphism group $\mathrm{Aut}(A)$ for each pre-Lie algebra $A$ would simplify the Nijenhuis family lists. This is a natural but non-trivial next step.  
  \item The Nijenhuis operators studied here can define para-complex structures on pre-Lie algebras (\cite{Wang}), analogous to classical complex geometry. This enables the study of para-complex quadratic pre-Lie algebras and para-complex pseudo-Hessian pre-Lie algebras, which may have rich geometric interpretations and connections to information geometry and statistical manifolds.  
  \item The algebraic framework in this paper is independent of dimension and applies to infinite-dimensional pre-Lie algebras. Promising examples include those from vertex operator algebras (VOAs) in conformal field theory (\cite[Subsection 2.6, 2.7]{Bai-in}). Such extensions could link our construction to theoretical physics and integrable systems.  
\end{itemize}

\textbf{Acknowledgements.}  This work is supported by the National Natural Science Foundation of China (Grant No. 12571042) and the Natural Science Foundation of Zhejiang Province (Grant No. LZ25A010002). The authors would like to thank Shizhuo Yu for
helpful discussions.  We are also grateful to the anonymous reviewers for their valuable comments and suggestions.

{\bf Data availability.} Data sharing not applicable to this article as no datasets were
generated or analysed during the current study.

{\bf Conflict of Interest.} The authors have no relevant financial or non-financial
interests to disclose.

\newpage

\section{Appendix}

\begin{table}[htbp]
    \centering
    \tiny
    \caption{Nijenhuis operators on 2-dimensional pre-Lie algebras.}
    \begin{tabular}{|c|c|>{\raggedright\arraybackslash}m{3cm}|c|l|l|}
        \hline
        \textbf{Pre-Lie} & \textbf{Dim.} & \textbf{Multiplication table} & \textbf{Commutative?} & \textbf{Nijenhuis operators} & \textbf{Parameters} \\
        \hline
        $A_1$ & 2 & $\begin{cases} e_1e_1=e_1 \\ e_2e_2=e_2 \end{cases}$ & Yes & 
        $N_{A_1}^1=\begin{pmatrix} n_{11} & 0 \\ n_{21} & n_{11}+n_{21} \end{pmatrix}$ & $n_{21}\neq0$ \\
        & & & & $N_{A_1}^2=\begin{pmatrix} n_{11} & 0 \\ 0 & n_{22} \end{pmatrix}$ & $/$ \\
        & & & & $N_{A_1}^3=\begin{pmatrix} n_{11} & n_{12} \\ 0 & n_{11}-n_{12} \end{pmatrix}$ & $n_{12}\neq0$ \\
        \hline
        $A_2$ & 2 & $\begin{cases}e_1e_1=e_1\\ e_1e_2=e_2e_1=e_2\end{cases}$ & Yes & 
        $N_{A_2}^1=\begin{pmatrix} n_{11} & 0 \\ 0 & n_{22} \end{pmatrix}$ & $/$ \\
        & & & & $N_{A_2}^2=\begin{pmatrix} n_{11} & n_{12} \\ 0 & n_{11} \end{pmatrix}$ & $n_{12}\neq0$ \\
        \hline
        $A_3$ & 2 & $ e_1e_1=e_1 $ & Yes & 
        $N_{A_3}^1=\begin{pmatrix} n_{11} & 0 \\ 0 & n_{22} \end{pmatrix}$ & $/$ \\
        & & & & $N_{A_3}^2=\begin{pmatrix} n_{11} & n_{12} \\ 0 & n_{11} \end{pmatrix}$ & $n_{12}\neq0$ \\
        \hline
        $A_4$ & 2 & $ e_ie_j=0,\; i,j=1,2 $ & Yes & 
        $N_{A_4}^1=\begin{pmatrix} n_{11} & n_{12} \\ n_{21} & n_{22} \end{pmatrix}$ & $/$ \\
        \hline
        $A_5$ & 2 & $e_1e_1=e_2 $ & No & 
        $N_{A_5}^1=\begin{pmatrix} n_{11} & n_{12} \\ 0 & n_{11} \end{pmatrix}$ & $/$ \\
        \hline
        $B_1$ & 2 & $\begin{cases} e_2e_1=-e_1\\ e_2e_2=e_1-e_2 \end{cases}$ & No & 
        $N_{B_1}^1=\begin{pmatrix} n_{11} & 0 \\ n_{21} & n_{11} \end{pmatrix}$ & $/$ \\
        \hline
        $B_2$ & 2 & $\begin{cases} e_2e_1=-e_1\\ e_2e_2=-e_2 \end{cases}$ & No & 
        $N_{B_2}^1=\begin{pmatrix} n_{11} & n_{12} \\ n_{21} & n_{22} \end{pmatrix}$ & $/$ \\
        \hline
        $B_3$ & 2 & $\begin{cases} e_2e_1=-e_1\\e_2e_2=ke_2,\; k\neq -1 \end{cases}$ & No & 
        $N_{B_3}^1=\begin{pmatrix} n_{11} & 0 \\ n_{21} & n_{11} \end{pmatrix}$ & $n_{21}\neq0$ \\
        & & & & $N_{B_3}^2=\begin{pmatrix} n_{11} & 0 \\ 0 & n_{22} \end{pmatrix}$ & $/$ \\
        \hline
        $B_4$ & 2 & $\begin{cases} e_1e_2=e_1\\e_2e_2=e_2 \end{cases}$ & No & 
        $N_{B_4}^1=\begin{pmatrix} n_{11} & n_{12} \\ n_{21} & n_{22} \end{pmatrix}$ & $/$ \\
        \hline
        $B_5$ & 2 & $\begin{cases} e_1e_2=le_1\\ e_2e_1=(l-1)e_1 \\ e_2e_2=e_1+le_2,\; l\neq 0 \end{cases}$ & No & 
        \makecell[l]{$l=1:$ \\ $N_{B_5}^1=\begin{pmatrix} n_{11} & 0 \\ n_{21} & n_{11} \end{pmatrix}$} & $/$ \\
        & & & & \makecell[l]{$l\neq 1:$ \\ $N_{B_5}^2=\begin{pmatrix} n_{11} & 0 \\ n_{21} & n_{22} \end{pmatrix}$} & $n_{11}=n_{22}$ \\
        & & & & \makecell[l]{$\phantom{l\neq 1:}$ \\ $N_{B_5}^3=\begin{pmatrix} n_{11} & 0 \\ \frac{n_{11}-n_{22}}{l-1} & n_{22} \end{pmatrix}$} & $n_{11}\neq n_{22}$ \\
        \hline
        $B_6$ & 2 & $\begin{cases} e_1e_1=2e_1\\e_1e_2=e_2\\e_2e_2=e_1 \end{cases}$ & No & 
        $N_{B_6}^1=\begin{pmatrix} n_{11} & 0 \\ n_{21} & n_{11}\pm n_{21}\sqrt{-1} \end{pmatrix}$ & $/$ \\
        & & & & $N_{B_6}^2=\begin{pmatrix} n_{11} & n_{12} \\ -n_{12} & n_{11} \end{pmatrix}$ & $n_{12}\neq0$ \\
        \hline
    \end{tabular}
    \label{tab:prelie_algebras_nijenhuis}
\end{table}
\newpage

\begin{table}[H]
    \centering
    \tiny
    \caption{Nijenhuis operators on 3-dimensional associative algebras.}
    \setlength{\tabcolsep}{1.5pt}
    \begin{tabular}{|c|c|l|l|l|l|}
    \hline
    \textbf{Ass} & \textbf{Dim.} & \textbf{Multiplication Table} & \textbf{Commutative?} & \textbf{Nijenhuis operators} ($C_1$-$C_5$) & \textbf{Parameters} \\
    \hline
    $C_1$ & 3 & $e_3 e_3 = e_1$ & Yes & 
    $N_{C_1}^1 = \begin{pmatrix} n_{11} & 0 & 0 \\ n_{21} & n_{22} & 0 \\ n_{31} & n_{32} & n_{11} \end{pmatrix}$ & $/$ \\
    & & & & 
    $N_{C_1}^2 = \begin{pmatrix} n_{11} & n_{12} & 0 \\ -\frac{(n_{11}-n_{33})^2}{n_{12}} & 2n_{33}-n_{11} & 0 \\ n_{31} & n_{32} & n_{33} \end{pmatrix}$ & $n_{12}\neq0$ \\
    \hline
    $C_2$ & 3 & $e_i e_j =0, i,j=1,2,3$ & Yes & 
    $N_{C_2}^1 = \begin{pmatrix} n_{11} & n_{12} & n_{13} \\ n_{21} & n_{22} & n_{23} \\ n_{31} & n_{32} & n_{33} \end{pmatrix}$ & $/$ \\
    \hline
    \multirow{9}{*}{$C_3$} & \multirow{9}{*}{3} & \multirow{9}{*}{$\begin{cases} e_2 e_2 = e_1 \\ e_3 e_3 = e_1 \end{cases}$} & \multirow{9}{*}{Yes} & 
    $N_{C_3}^1 = \begin{pmatrix} n_{11} & 0 & 0 \\ n_{21} & n_{11} & 0 \\ n_{31} & 0 & n_{11} \end{pmatrix}$ & $/$\\
    \cline{5-6}
    & & & &  $N_{C_3}^2 = \begin{pmatrix} n_{11} & 0 & 0 \\ n_{21} & n_{11} & 0 \\ n_{31} & n_{32} & n_{33} \end{pmatrix}$& \makecell[l]{$n_{33}=n_{11}\pm n_{32}\sqrt{-1},$\\$n_{32}\neq0$}  \\  
    \cline{5-6}
    & & && $N_{C_3}^3 = \begin{pmatrix} n_{11} & 0 & 0 \\ n_{21} & n_{22} & n_{23} \\ n_{31} & 0 & n_{11} \end{pmatrix}$& \makecell[l]{$n_{22}=n_{11}\pm n_{23}\sqrt{-1},$\\$n_{23}\neq0$}  \\  
    \cline{5-6}
    & & & & $N_{C_3}^4 = \begin{pmatrix} n_{11} & 0 & 0 \\ n_{21} & n_{22} & n_{23} \\ n_{31} & n_{32} & n_{33} \end{pmatrix}$ & \begin{tabular}{@{}l@{}}$n_{23}\neq0$,$n_{32}\neq0$\\$n_{22}=n_{11}-n_{23}\sqrt{-1}$,\\$n_{33}=n_{11}+n_{32}\sqrt{-1}$\end{tabular} \\
\cline{5-6}
    & & & & $N_{C_3}^5 = \begin{pmatrix} n_{11} & 0 & 0 \\ n_{21} & n_{22} & n_{23} \\ n_{31} & n_{32} & n_{33} \end{pmatrix}$ & \begin{tabular}{@{}l@{}}$n_{23}\neq0$,$n_{32}\neq0$\\$n_{22}=n_{11}+n_{23}\sqrt{-1}$,\\$n_{33}=n_{11}-n_{32}\sqrt{-1}$\end{tabular}  \\
    \hline
    $C_4$ & 3 & $\begin{cases} e_2 e_3 = e_3 e_2 = e_1 \\ e_3 e_3 = e_2 \end{cases}$ & Yes & 
    $N_{C_4}^1 = \begin{pmatrix} n_{11} & 0 & 0 \\ n_{21} & n_{11} & 0 \\ n_{31} & n_{32} & n_{11} \end{pmatrix}$ & $/$ \\
    \hline
    \multirow{16}{*}{$C_5$} & \multirow{16}{*}{3} & \multirow{16}{*}{$\begin{cases} e_1 e_1 = e_1 \\ e_2 e_2 = e_2 \\ e_3 e_3 = e_3 \end{cases}$} & \multirow{16}{*}{Yes} & 
    $N_{C_5}^1 = \text{diag}(n_{11}, n_{22}, n_{33})$ & $/$ \\
    & & & & $N_{C_5}^2 = \begin{pmatrix} n_{11} & 0 & 0 \\ 0 & n_{22} & 0 \\ n_{31} & 0 & n_{11}+n_{31} \end{pmatrix}$ & $n_{31}\neq0$ \\
    & & & & $N_{C_5}^3 = \begin{pmatrix} n_{11} & 0 & 0 \\ 0 & n_{22} & 0 \\ 0 & n_{32} & n_{22}+n_{32} \end{pmatrix}$ & $n_{32}\neq0$ \\
    & & & & $N_{C_5}^4 = \begin{pmatrix} n_{11} & 0 & 0 \\ n_{21} & n_{11}+n_{21} & 0 \\ 0 & 0 & n_{33} \end{pmatrix}$ & $n_{21}\neq0$ \\
    & & & & $N_{C_5}^5 = \begin{pmatrix} n_{11} & 0 & 0 \\ 0 & n_{23}+n_{33} & n_{23} \\ 0 & 0 & n_{33} \end{pmatrix}$ & $n_{23}\neq0$ \\
    & & & & $N_{C_5}^6 = \begin{pmatrix} n_{11} & 0 & 0 \\ n_{21} & n_{11}+n_{21} & 0 \\ n_{32} & n_{32} & n_{11}+n_{21}+n_{32} \end{pmatrix}$ & $n_{21},n_{32}\neq0$ \\
    & & & & $N_{C_5}^7 = \begin{pmatrix} n_{22}+n_{32}-n_{31} & 0 & 0 \\ 0 & n_{22} & 0 \\ n_{31} & n_{32} & n_{22}+n_{32} \end{pmatrix}$ & $n_{31},n_{32}\neq0$ \\
    & & & & $N_{C_5}^8 = \begin{pmatrix} n_{11} & 0 & n_{13} \\ 0 & n_{22} & 0 \\ 0 & 0 & n_{11}+n_{13} \end{pmatrix}$ & $n_{13}\neq0$ \\
    & & & & $N_{C_5}^9 = \begin{pmatrix} n_{11} & n_{12} & 0 \\ 0 & n_{11}-n_{12} & 0 \\ 0 & 0 & n_{33} \end{pmatrix}$ & $n_{12}\neq0$ \\
    & & & & $N_{C_5}^{10} = \begin{pmatrix} n_{11} & 0 & n_{13} \\ n_{23} & n_{11}+n_{23} & n_{23} \\ 0 & 0 & n_{11}-n_{13} \end{pmatrix}$ & $n_{13},n_{23}\neq0$ \\
    & & & & $N_{C_5}^{11} = \begin{pmatrix} n_{11} & n_{12} & 0 \\ 0 & n_{11}-n_{12} & 0 \\ n_{31} & n_{31} & n_{31}+n_{11} \end{pmatrix}$ & $n_{12},n_{31}\neq0$ \\
       & & & & $N_{C_5}^{12} = \begin{pmatrix} n_{23}+n_{33}-n_{21} & 0 & 0 \\ n_{21} & n_{23}+n_{33} & n_{23} \\ 0 & 0 & n_{33} \end{pmatrix}$ & $n_{21},n_{23}\neq0$ \\
  & & & & $N_{C_5}^{13} = \begin{pmatrix} n_{22}-n_{23}-n_{31} & 0 & 0 \\ n_{23} & n_{22} & n_{23} \\ n_{31} & 0 & n_{22}-n_{23} \end{pmatrix}$ & $n_{31},n_{23}\neq0$ \\
    \hline
    \end{tabular}
    \label{tab:nijenhuis_3d_associative}
\end{table}
\begin{table}[H]
    \centering
    \tiny
    \caption{Nijenhuis operators on 3-dimensional associative algebras (continued).}
    \setlength{\tabcolsep}{1.5pt}
    \begin{tabular}{|c|c|l|l|l|l|}
    \hline
    \textbf{Ass} & \textbf{Dim.} & \textbf{Multiplication Table} & \textbf{Commutative?} & \textbf{Nijenhuis operators} ($C_5$-$C_8$) & \textbf{Parameters} \\
    \hline
    \multirow{3}{*}{$C_5$} & \multirow{3}{*}{3} & \multirow{3}{*}{$\begin{cases} e_1 e_1 = e_1 \\ e_2 e_2 = e_2 \\ e_3 e_3 = e_3 \end{cases}$} & \multirow{3}{*}{Yes} & 
    $N_{C_5}^{14} = \begin{pmatrix} n_{11} & n_{12} & n_{13} \\ 0 & n_{11}-n_{12} & 0 \\ 0 & 0 & n_{11}-n_{13} \end{pmatrix}$ & $n_{12},n_{13}\neq0$ \\
    & & & & $N_{C_5}^{15} = \begin{pmatrix} n_{11} & n_{12} & n_{12} \\ 0 & n_{11}-n_{12}-n_{32} & 0 \\ 0 & n_{32} & n_{11}-n_{12} \end{pmatrix}$ & $n_{12},n_{32}\neq0$ \\
    & & & & $N_{C_5}^{16} = \begin{pmatrix} n_{11} & n_{12} & n_{12} \\ 0 & n_{11}-n_{12} & n_{23} \\ 0 & 0 & n_{11}-n_{12}-n_{23} \end{pmatrix}$ & $n_{12},n_{23}\neq0$ \\
    \hline
    \multirow{8}{*}{$C_6$} & \multirow{8}{*}{3} & \multirow{8}{*}{$\begin{cases} e_2 e_2 = e_2 \\ e_3 e_3 = e_3 \end{cases}$} & \multirow{8}{*}{Yes} & 
    $N_{C_6}^1 = \text{diag}(n_{11}, n_{22}, n_{33})$ & $/$ \\
    & & & & $N_{C_6}^2 = \begin{pmatrix} n_{11} & 0 & 0 \\ 0 & n_{22} & 0 \\ n_{31} & 0 & n_{11} \end{pmatrix}$ & $n_{31}\neq0$ \\
    & & & & $N_{C_6}^3 = \begin{pmatrix} n_{11} & 0 & 0 \\ n_{21} & n_{11} & 0 \\ 0 & 0 & n_{33} \end{pmatrix}$ & $n_{21}\neq0$ \\
    & & & & $N_{C_6}^4 = \begin{pmatrix} n_{11} & 0 & 0 \\ n_{21} & n_{11} & 0 \\ n_{31} & 0 & n_{11} \end{pmatrix}$ & $n_{21},n_{31}\neq0$ \\
    & & & & $N_{C_6}^5 = \begin{pmatrix} n_{11} & 0 & 0 \\ 0 & n_{22} & 0 \\ 0 & n_{32} & n_{22}+n_{32} \end{pmatrix}$ & $n_{32}\neq0$ \\
    & & & & $N_{C_6}^6 = \begin{pmatrix} n_{11} & 0 & 0 \\ 0 & n_{23}+n_{33} & n_{23} \\ 0 & 0 & n_{33} \end{pmatrix}$ & $n_{23}\neq0$ \\
    & & & & $N_{C_6}^7 = \begin{pmatrix} n_{22}+n_{32} & 0 & 0 \\ 0 & n_{22} & 0 \\ n_{31} & n_{32} & n_{22}+n_{32} \end{pmatrix}$ & $n_{31},n_{32}\neq0$ \\
    & & & & $N_{C_6}^8 = \begin{pmatrix} n_{23}+n_{33} & 0 & 0 \\ n_{21} & n_{23}+n_{33} & n_{23} \\ 0 & 0 & n_{33} \end{pmatrix}$ & $n_{21},n_{23}\neq0$ \\
    \hline
    \multirow{8}{*}{$C_7$} & \multirow{8}{*}{3} & \multirow{8}{*}{$\begin{cases} e_1 e_3 = e_3 e_1 = e_1 \\ e_2 e_2 = e_2 \\ e_3 e_3 = e_3 \end{cases}$} & \multirow{8}{*}{Yes} & 
    $N_{C_7}^1 = \text{diag}(n_{11}, n_{22}, n_{33})$ & $/$ \\
    & & & & $N_{C_7}^2 = \begin{pmatrix} n_{11} & 0 & 0 \\ 0 & n_{22} & 0 \\ n_{31} & 0 & n_{11} \end{pmatrix}$ & $n_{31}\neq0$ \\
    & & & & $N_{C_7}^3 = \begin{pmatrix} n_{11} & 0 & 0 \\ n_{21} & n_{11} & 0 \\ n_{31} & 0 & n_{11} \end{pmatrix}$ & $n_{21}\neq0$ \\
    & & & & $N_{C_7}^4 = \begin{pmatrix} n_{22}+n_{32} & 0 & 0 \\ 0 & n_{22} & 0 \\ n_{31} & n_{32} & n_{22}+n_{32} \end{pmatrix}$ & $n_{31},n_{32}\neq0$ \\
    & & & & $N_{C_7}^5 = \begin{pmatrix} n_{11} & 0 & 0 \\ 0 & n_{22} & 0 \\ 0 & n_{32} & n_{22}+n_{32} \end{pmatrix}$ & $n_{32}\neq0$ \\
    & & & & $N_{C_7}^6 = \begin{pmatrix} n_{11} & 0 & 0 \\ 0 & n_{23}+n_{33} & n_{23} \\ 0 & 0 & n_{33} \end{pmatrix}$ & $n_{23}\neq0$ \\
    & & & & $N_{C_7}^7 = \begin{pmatrix} n_{22} & 0 & 0 \\ n_{21} & n_{22} & 0 \\ -n_{21} & n_{32} & n_{22}+n_{32} \end{pmatrix}$ & $n_{21},n_{32}\neq0$ \\
    & & & & $N_{C_7}^8 = \begin{pmatrix} n_{11} & 0 & 0 \\ n_{21} & n_{23}+n_{11} & n_{23} \\ -n_{21} & 0 & n_{11} \end{pmatrix}$ & $n_{21},n_{23}\neq0$ \\
    \hline
    \multirow{3}{*}{$C_8$} & \multirow{3}{*}{3} & \multirow{3}{*}{$e_3 e_3 = e_3$} & \multirow{3}{*}{Yes} & 
    $N_{C_8}^1 = \begin{pmatrix} n_{11} & n_{12} & 0 \\ n_{21} & n_{22} & 0 \\ 0 & 0 & n_{33} \end{pmatrix}$ & $/$ \\
    & & & & $N_{C_8}^2 = \begin{pmatrix} n_{11} & n_{12} & 0 \\ 0 & n_{22} & 0 \\ 0 & n_{32} & n_{22} \end{pmatrix}$ & $n_{32}\neq0$ \\
    & & & & $N_{C_8}^3 = \begin{pmatrix} n_{33}-\frac{n_{21}n_{32}}{n_{31}} & \frac{n_{32}n_{33}-n_{22}n_{32}}{n_{31}} & 0 \\ n_{21} & n_{22} & 0 \\ n_{31} & n_{32} & n_{33} \end{pmatrix}$ & $n_{31}\neq0$ \\
    \hline
    \end{tabular}
    \label{tab:nijenhuis_3d_associative_complete}
\end{table}

\begin{table}[H]
    \centering
    \tiny
    \caption{Nijenhuis operators on 3-dimensional associative algebras (continued).}
    \setlength{\tabcolsep}{1.5pt}
    \begin{tabular}{|c|c|l|l|l|l|}
    \hline
    \textbf{Ass} & \textbf{Dim.} & \textbf{Multiplication Table} & \textbf{Commutative?} & \textbf{Nijenhuis operators} ($C_9$-$C_{12}$, $D_1$) & \textbf{Parameters} \\
    \hline
    \multirow{6}{*}{$C_9$} & \multirow{6}{*}{3} & \multirow{6}{*}{$\begin{cases} e_1 e_3 = e_3 e_1 = e_1 \\ e_3 e_3 = e_3 \end{cases}$} & \multirow{6}{*}{Yes} & 
    $N_{C_9}^1 = \text{diag}(n_{11}, n_{22}, n_{33})$ & $/$ \\
    & & & & $N_{C_9}^2 = \begin{pmatrix} n_{11} & 0 & 0 \\ 0 & n_{22} & 0 \\ 0 & n_{32} & n_{22} \end{pmatrix}$ & $n_{32}\neq0$ \\
    & & & & $N_{C_9}^3 = \begin{pmatrix} n_{11} & 0 & 0 \\ 0 & n_{22} & 0 \\ n_{31} & 0 & n_{11} \end{pmatrix}$ & $n_{31}\neq0$ \\
    & & & & $N_{C_9}^4 = \begin{pmatrix} n_{11} & 0 & 0 \\ 0 & n_{11} & 0 \\ n_{31} & n_{32} & n_{11} \end{pmatrix}$ & $n_{31},n_{32}\neq0$ \\
    & & & & $N_{C_9}^5 = \begin{pmatrix} n_{11} & 0 & 0 \\ n_{21} & n_{22} & 0 \\ n_{31} & 0 & n_{11} \end{pmatrix}$ & $n_{21}\neq0$ \\
    & & & & $N_{C_9}^6 = \begin{pmatrix} n_{11} & n_{12} & 0 \\ 0 & n_{22} & 0 \\ 0 & n_{32} & n_{22} \end{pmatrix}$ & $n_{12}\neq0$ \\
    \hline
    \multirow{3}{*}{$C_{10}$} & \multirow{3}{*}{3} & \multirow{3}{*}{$\begin{cases} e_1 e_3 = e_3 e_1 = e_1 \\ e_2 e_3 = e_3 e_2 = e_2 \\ e_3 e_3 = e_3 \end{cases}$} & \multirow{3}{*}{Yes} & 
    $N_{C_{10}}^1 = \begin{pmatrix} n_{11} & n_{12} & 0 \\ n_{21} & n_{22} & 0 \\ 0 & 0 & n_{33} \end{pmatrix}$ & $/$ \\
    & & & & $N_{C_{10}}^2 = \begin{pmatrix} n_{11} & n_{12} & 0 \\ 0 & n_{22} & 0 \\ 0 & n_{32} & n_{22} \end{pmatrix}$ & $n_{32}\neq0$ \\
    & & & & $N_{C_{10}}^3 = \begin{pmatrix} n_{33}-\frac{n_{21}n_{32}}{n_{31}} & \frac{n_{32}n_{33}-n_{22}n_{32}}{n_{31}} & 0 \\ n_{21} & n_{22} & 0 \\ n_{31} & n_{32} & n_{33} \end{pmatrix}$ & $n_{31}\neq0$ \\
    \hline
    \multirow{3}{*}{$C_{11}$} & \multirow{3}{*}{3} & \multirow{3}{*}{$\begin{cases} e_1 e_1 = e_2 \\ e_3 e_3 = e_3 \end{cases}$} & \multirow{3}{*}{Yes} & 
    $N_{C_{11}}^1 = \begin{pmatrix} n_{11} & n_{12} & 0 \\ 0 & n_{11} & 0 \\ 0 & 0 & n_{33} \end{pmatrix}$ & $/$ \\
    & & & & $N_{C_{11}}^2 = \begin{pmatrix} n_{11} & n_{12} & 0 \\ 0 & n_{11} & 0 \\ 0 & n_{32} & n_{11} \end{pmatrix}$ & $n_{32}\neq0$ \\
    & & & & $N_{C_{11}}^3 = \begin{pmatrix} n_{11} & -n_{31} & 0 \\ 0 & n_{11} & 0 \\ n_{31} & n_{32} & n_{11} \end{pmatrix}$ & $n_{31}\neq0$ \\
    \hline
    \multirow{3}{*}{$C_{12}$} & \multirow{3}{*}{3} & \multirow{3}{*}{$\begin{cases} e_1 e_1 = e_2 \\ e_1 e_3 = e_3 e_1 = e_1 \\ e_2 e_3 = e_3 e_2 = e_2 \\ e_3 e_3 = e_3 \end{cases}$} & \multirow{3}{*}{Yes} & 
    $N_{C_{12}}^1 = \begin{pmatrix} n_{11} & n_{12} & 0 \\ 0 & n_{11} & 0 \\ 0 & 0 & n_{33} \end{pmatrix}$ & $/$ \\
    & & & & $N_{C_{12}}^2 = \begin{pmatrix} n_{11} & n_{12} & 0 \\ 0 & n_{11} & 0 \\ 0 & n_{32} & n_{11} \end{pmatrix}$ & $n_{32}\neq0$ \\
    & & & & $N_{C_{12}}^3 = \begin{pmatrix} n_{11} & n_{31} & 0 \\ 0 & n_{11} & 0 \\ n_{31} & n_{32} & n_{11} \end{pmatrix}$ & $n_{31}\neq0$ \\
    \hline
    \multirow{3}{*}{$D_1$} & \multirow{3}{*}{3} & \multirow{3}{*}{$\begin{cases} e_1 \cdot e_2 = \frac{1}{2} e_3 \\ e_2 \cdot e_1 = -\frac{1}{2} e_3 \end{cases}$} & \multirow{3}{*}{No} & 
    $N_{D_1}^1 = \begin{pmatrix} n_{11} & 0 & n_{13} \\ n_{21} & n_{22} & n_{23} \\ 0 & 0 & n_{11} \end{pmatrix}$ & $/$ \\
    & & & & $N_{D_1}^2 = \begin{pmatrix} n_{11} & n_{12} & n_{13} \\ 0 & n_{22} & n_{23} \\ 0 & 0 & n_{11} \end{pmatrix}$ & $n_{12} \neq 0$ \\
    & & & & $N_{D_1}^3 = \begin{pmatrix} n_{11} & n_{12} & n_{13} \\ n_{21} & n_{33} - \frac{n_{12}n_{21}}{n_{33}-n_{11}} & n_{23} \\ 0 & 0 & n_{33} \end{pmatrix}$ & $n_{33} \neq n_{11}$ \\
    \hline
    \end{tabular}
    \label{tab:nijenhuis_3d_associative_C9-C12}
\end{table}

\begin{table}[H]
    \centering
    \tiny
    \caption{Nijenhuis operators on 3-dimensional associative algebras (continued).}
    \setlength{\tabcolsep}{1.5pt}
    \begin{tabular}{|c|c|l|c|l|l|}
    \hline
    \textbf{Ass} & \textbf{Dim.} & \textbf{Multiplication Table} & \textbf{Commutative?} & \textbf{Nijenhuis operators $D_2$-$D_6$} & \textbf{Parameters} \\
    \hline
    \multirow{4}{*}{$D_2$} & \multirow{4}{*}{3} & \multirow{4}{*}{$e_2 \cdot e_1 = - e_3$} & \multirow{4}{*}{No} & 
    $N_{D_2}^1 = \begin{pmatrix} n_{11} & 0 & n_{13} \\ n_{21} & n_{22} & n_{23} \\ 0 & 0 & n_{22} \end{pmatrix}$ & $n_{21} \neq 0$ \\
    & & & & $N_{D_2}^2 = \begin{pmatrix} n_{11} & n_{12} & n_{13} \\ 0 & n_{22} & n_{23} \\ 0 & 0 & n_{11} \end{pmatrix}$ & $n_{12} \neq 0$ \\
    & & & & $N_{D_2}^3 = \begin{pmatrix} n_{11} & 0 & n_{13} \\ 0 & n_{22} & n_{23} \\ 0 & 0 & n_{22} \end{pmatrix}$ & $/$ \\
    & & & & $N_{D_2}^4 = \begin{pmatrix} n_{33} & 0 & n_{13} \\ 0 & n_{22} & n_{23} \\ 0 & 0 & n_{33} \end{pmatrix}$ & $n_{33} \neq n_{22}$ \\
    \hline
    \multirow{3}{*}{$D_3$} & \multirow{3}{*}{3} & \multirow{3}{*}{$\begin{cases} e_1 \cdot e_1 = e_3 \\ e_1 \cdot e_2 = e_3 \\ e_2 \cdot e_2 = \lambda e_3, \ \lambda \neq 0 \end{cases}$} & \multirow{3}{*}{No} & 
    $N_{D_3}^1 = \begin{pmatrix} n_{11} & n_{12} & n_{13} \\ n_{21} & n_{22} & n_{23} \\ 0 & 0 & n_{33} \end{pmatrix}$ & \begin{tabular}{@{}l@{}}$n_{12} \neq 0$,\\$n_{11} = n_{33} + \frac{(-1+\sqrt{1-4\lambda})n_{12}}{2}$,\\$n_{21} = \frac{(-1+\sqrt{1-4\lambda})(n_{22}-n_{33})}{2}$\end{tabular} \\
    \cline{5-6}
    & & & & $N_{D_3}^2 = \begin{pmatrix} n_{11} & n_{12} & n_{13} \\ n_{21} & n_{22} & n_{23} \\ 0 & 0 & n_{33} \end{pmatrix}$ & \begin{tabular}{@{}l@{}}$n_{12} \neq 0$,\\$n_{11} = n_{33} + \frac{(-1-\sqrt{1-4\lambda})n_{12}}{2}$,\\$n_{21} = \frac{(-1-\sqrt{1-4\lambda})(n_{22}-n_{33})}{2}$\end{tabular} \\
    \cline{5-6}
    & & & & $N_{D_3}^3 = \begin{pmatrix} n_{33} & 0 & n_{13} \\ n_{21} & n_{22} & n_{23} \\ 0 & 0 & n_{33} \end{pmatrix}$ & $n_{21} = \frac{(-1\pm \sqrt{1-4\lambda})(n_{22}-n_{33})}{2}$ \\
    \hline
    \multirow{4}{*}{$D_4$} & \multirow{4}{*}{3} & \multirow{4}{*}{$\begin{cases} e_3 \cdot e_2 = e_2 \\ e_3 \cdot e_3 = e_3 \end{cases}$} & \multirow{4}{*}{No} & 
    $N_{D_4}^1 = \begin{pmatrix} n_{11} & 0 & 0 \\ 0 & n_{22} & n_{23} \\ 0 & n_{32} & n_{33} \end{pmatrix}$ & $/$ \\
    & & & & $N_{D_4}^2 = \begin{pmatrix} n_{11} & 0 & 0 \\ 0 & n_{22} & 0 \\ n_{31} & n_{32} & n_{11} \end{pmatrix}$ & $n_{31} \neq 0$ \\
    & & & & $N_{D_4}^3 = \begin{pmatrix} n_{11} & 0 & 0 \\ n_{21} & n_{22} & 0 \\ n_{31} & n_{32} & n_{11} \end{pmatrix}$ & $n_{21} \neq 0$ \\
    & & & & $N_{D_4}^4 = \begin{pmatrix} n_{11} & n_{12} & 0 \\ 0 & n_{22} & 0 \\ 0 & n_{32} & n_{22} \end{pmatrix}$ & $n_{12} \neq 0$ \\
    \hline
    \multirow{4}{*}{$D_5$} & \multirow{4}{*}{3} & \multirow{4}{*}{$\begin{cases} e_2 \cdot e_3 = e_2 \\ e_3 \cdot e_3 = e_3 \end{cases}$} & \multirow{4}{*}{No} & 
    $N_{D_5}^1 = \begin{pmatrix} n_{11} & 0 & 0 \\ 0 & n_{22} & n_{23} \\ 0 & n_{32} & n_{33} \end{pmatrix}$ & $/$ \\
    & & & & $N_{D_5}^2 = \begin{pmatrix} n_{11} & 0 & 0 \\ 0 & n_{22} & 0 \\ n_{31} & n_{32} & n_{11} \end{pmatrix}$ & $n_{31} \neq 0$ \\
    & & & & $N_{D_5}^3 = \begin{pmatrix} n_{11} & 0 & 0 \\ n_{21} & n_{22} & 0 \\ n_{31} & n_{32} & n_{11} \end{pmatrix}$ & $n_{21} \neq 0$ \\
    & & & & $N_{D_5}^4 = \begin{pmatrix} n_{11} & n_{12} & 0 \\ 0 & n_{22} & 0 \\ 0 & n_{32} & n_{22} \end{pmatrix}$ & $n_{12} \neq 0$ \\
    \hline
    \multirow{6}{*}{$D_6$} & \multirow{6}{*}{3} & \multirow{6}{*}{$\begin{cases} e_1 \cdot e_1 = e_1 \\ e_3 \cdot e_2 = e_2 \\ e_3 \cdot e_3 = e_3 \end{cases}$} & \multirow{6}{*}{No} & 
    $N_{D_6}^1 = \begin{pmatrix} n_{11} & 0 & 0 \\ 0 & n_{22} & n_{23} \\ 0 & n_{32} & n_{33} \end{pmatrix}$ & $/$ \\
    & & & & $N_{D_6}^2 = \begin{pmatrix} n_{11} & n_{12} & 0 \\ 0 & n_{11} & 0 \\ 0 & n_{32} & n_{11} \end{pmatrix}$ & $n_{12} \neq 0$ \\
    & & & & $N_{D_6}^3 = \begin{pmatrix} n_{11} & 0 & 0 \\ n_{21} & n_{22} & n_{21} \\ n_{31} & n_{32} & n_{11}+n_{31} \end{pmatrix}$ & $n_{21} \neq 0$ \\
    & & & & $N_{D_6}^4 = \begin{pmatrix} n_{11} & 0 & n_{13} \\ 0 & n_{22} & 0 \\ 0 & 0 & n_{11}-n_{13} \end{pmatrix}$ & $n_{13} \neq 0$ \\
    & & & & $N_{D_6}^5 = \begin{pmatrix} n_{11} & 0 & 0 \\ 0 & n_{22} & 0 \\ n_{31} & n_{32} & n_{11}+n_{31} \end{pmatrix}$ & $n_{31} \neq 0$ \\
    \cline{5-6}
    & & & & $N_{D_6}^6 = \begin{pmatrix} n_{11} & n_{12} & n_{13} \\ 0 & n_{22} & 0 \\ 0 & n_{32}& n_{11}-n_{13} \end{pmatrix}$ & \makecell[l]{
    $n_{12} \neq 0,\ n_{13} \neq 0,$ \\
    $n_{32}=\frac{n_{11}n_{12}-n_{12}n_{13}-n_{12}n_{22}}{n_{13}}$}
    \\
    \hline
    \end{tabular}
    \label{tab:nijenhuis_3d_associative_D1-D6}
\end{table}

\begin{table}[H]
    \centering
    \tiny
    \caption{Nijenhuis operators on 3-dimensional associative algebras (continued).}
    \setlength{\tabcolsep}{1.5pt}
    \begin{tabular}{|c|c|l|c|l|l|}
    \hline
    \textbf{Ass} & \textbf{Dim.} & \textbf{Multiplication Table} & \textbf{Commutative?} & \textbf{Nijenhuis operators} & \textbf{Parameters} \\
    \hline
    \multirow{6}{*}{$D_7$} & \multirow{6}{*}{3} & \multirow{6}{*}{$\begin{cases} e_1 \cdot e_1 = e_1 \\ e_2 \cdot e_3 = e_2 \\ e_3 \cdot e_3 = e_3 \end{cases}$} & \multirow{6}{*}{No} & 
    $N_{D_7}^1 = \begin{pmatrix} n_{11} & 0 & 0 \\ 0 & n_{22} & n_{23} \\ 0 & n_{32} & n_{33} \end{pmatrix}$ & $/$ \\
    & & & & $N_{D_7}^2 = \begin{pmatrix} n_{11} & n_{12} & 0 \\ 0 & n_{11} & 0 \\ 0 & n_{32} & n_{11} \end{pmatrix}$ & $n_{12} \neq 0$ \\
    & & & & $N_{D_7}^3 = \begin{pmatrix} n_{11} & 0 & 0 \\ n_{21} & n_{22} & n_{21} \\ n_{31} & n_{32} & n_{11}+n_{31} \end{pmatrix}$ & $n_{21} \neq 0$ \\
    & & & & $N_{D_7}^4 = \begin{pmatrix} n_{11} & 0 & n_{13} \\ 0 & n_{22} & 0 \\ 0 & 0 & n_{11}-n_{13} \end{pmatrix}$ & $n_{13} \neq 0$ \\
    & & & & $N_{D_7}^5 = \begin{pmatrix} n_{11} & 0 & 0 \\ 0 & n_{22} & 0 \\ n_{31} & n_{32} & n_{11}+n_{31} \end{pmatrix}$ & $n_{31} \neq 0$ \\
    & & & & $N_{D_7}^6 = \begin{pmatrix} n_{11} & n_{12} & n_{13} \\ 0 & n_{22} & 0 \\ 0 & \frac{n_{11}n_{12}-n_{12}n_{13}-n_{12}n_{22}}{n_{13}} & n_{11}-n_{13} \end{pmatrix}$ & $n_{12} \neq 0,\ n_{13} \neq 0$ \\
    \hline
    \multirow{4}{*}{$D_8$} & \multirow{4}{*}{3} & \multirow{4}{*}{$\begin{cases} e_1 \cdot e_3 = e_1 \\ e_3 \cdot e_1 = e_1 \\ e_3 \cdot e_2 = e_2 \\ e_3 \cdot e_3 = e_3 \end{cases}$} & \multirow{4}{*}{No} & 
    $N_{D_8}^1 = \begin{pmatrix} n_{11} & 0 & 0 \\ 0 & n_{22} & n_{23} \\ 0 & n_{32} & n_{33} \end{pmatrix}$ & $/$ \\
    & & & & $N_{D_8}^2 = \begin{pmatrix} n_{11} & 0 & 0 \\ 0 & n_{22} & 0 \\ n_{31} & n_{32} & n_{11} \end{pmatrix}$ & $n_{31} \neq 0$ \\
    & & & & $N_{D_8}^3 = \begin{pmatrix} n_{11} & 0 & 0 \\ n_{21} & n_{22} & 0 \\ n_{31} & n_{32} & n_{11} \end{pmatrix}$ & $n_{21} \neq 0$ \\
    & & & & $N_{D_8}^4 = \begin{pmatrix} n_{11} & n_{12} & 0 \\ 0 & n_{22} & 0 \\ 0 & n_{32} & n_{22} \end{pmatrix}$ & $n_{12} \neq 0$ \\
    \hline
    \multirow{5}{*}{$D_9$} & \multirow{5}{*}{3} & \multirow{5}{*}{$\begin{cases} e_1 \cdot e_1 = e_1 \\ e_1 \cdot e_2 = e_2 \cdot e_1 = e_2 \\ e_1 \cdot e_3 = e_3 \cdot e_1 = e_3 \\ e_3 \cdot e_2 = e_2 \\ e_3 \cdot e_3 = e_3 \end{cases}$} & \multirow{5}{*}{No} & 
    $N_{D_9}^1 = \begin{pmatrix} n_{11} & 0 & 0 \\ 0 & n_{22} & n_{23} \\ 0 & n_{32} & n_{33} \end{pmatrix}$ & $/$ \\
    & & & & $N_{D_9}^2 = \begin{pmatrix} n_{31}+n_{33} & 0 & 0 \\ 0 & n_{22} & n_{23} \\ n_{31} & n_{32} & n_{33} \end{pmatrix}$ & $n_{31} \neq 0$ \\
    & & & & $N_{D_9}^3 = \begin{pmatrix} n_{31}+n_{33} & 0 & 0 \\ n_{21} & n_{22} & -n_{21} \\ n_{31} & n_{32} & n_{33} \end{pmatrix}$ & $n_{21} \neq 0$ \\
    & & & & $N_{D_9}^4 = \begin{pmatrix} n_{11} & n_{12} & 0 \\ 0 & n_{11} & 0 \\ 0 & n_{32} & n_{11} \end{pmatrix}$ & $n_{12} \neq 0$ \\
    & & & & $N_{D_9}^5 = \begin{pmatrix} n_{11} & n_{12} & n_{13} \\ 0 & n_{22} & 0 \\ 0 & \frac{n_{11}n_{12}+n_{12}n_{13}-n_{12}n_{22}}{n_{13}} & n_{11}+n_{13} \end{pmatrix}$ & $n_{13} \neq 0$ \\
    \hline
    $D_{10}$ & 3 & $\begin{cases} e_3 \cdot e_1 = e_1 \\ e_3 \cdot e_2 = e_2 \\ e_3 \cdot e_3 = e_3 \end{cases}$ & No & 
    $N_{D_{10}}^1 = \begin{pmatrix} n_{11} & n_{12} & n_{13} \\ n_{21} & n_{22} & n_{23} \\ n_{31} & n_{32} & n_{33} \end{pmatrix}$ & $/$ \\
    \hline
    $D_{11}$ & 3 & $\begin{cases} e_1 \cdot e_3 = e_1 \\ e_2 \cdot e_3 = e_2 \\ e_3 \cdot e_3 = e_3 \end{cases}$ & No & 
    $N_{D_{11}}^1 = \begin{pmatrix} n_{11} & n_{12} & n_{13} \\ n_{21} & n_{22} & n_{23} \\ n_{31} & n_{32} & n_{33} \end{pmatrix}$ & $/$ \\
    \hline
    \multirow{5}{*}{$D_{12}$} & \multirow{5}{*}{3} & \multirow{5}{*}{$\begin{cases} e_3 \cdot e_1 = e_1 \\ e_2 \cdot e_3 = e_2 \\ e_3 \cdot e_3 = e_3 \end{cases}$} & \multirow{5}{*}{No} & 
    $N_{D_{12}}^1 = \begin{pmatrix} n_{11} & 0 & 0 \\ 0 & n_{22} & 0 \\ n_{31} & n_{32} & n_{33} \end{pmatrix}$ & $/$ \\
    & & & & $N_{D_{12}}^2 = \begin{pmatrix} n_{11} & 0 & 0 \\ n_{21} & n_{22} & 0 \\ n_{31} & n_{32} & n_{11} \end{pmatrix}$ & $n_{21} \neq 0$ \\
    & & & & $N_{D_{12}}^3 = \begin{pmatrix} n_{11} & n_{12} & 0 \\ 0 & n_{22} & 0 \\ n_{31} & n_{32} & n_{22} \end{pmatrix}$ & $n_{12} \neq 0$ \\
    & & & & $N_{D_{12}}^4 = \begin{pmatrix} n_{11} & 0 & 0 \\ n_{21} & n_{22} & n_{23} \\ \frac{n_{21}n_{33}-n_{11}n_{21}}{n_{23}} & n_{32} & n_{33} \end{pmatrix}$ & $n_{23} \neq 0$ \\
    & & & & $N_{D_{12}}^5 = \begin{pmatrix} n_{11} & n_{12} & n_{13} \\ 0 & n_{22} & 0 \\ n_{31} & \frac{n_{12}n_{33}-n_{12}n_{22}}{n_{13}} & n_{33} \end{pmatrix}$ & $n_{13} \neq 0$ \\
    \hline
    \end{tabular}
    \label{tab:nijenhuis_3d_associative_D7-D12}
\end{table}

\end{document}